%% file: PRX_main.tex
\documentclass[pra,superscriptaddress,notitlepage,twocolumn]{revtex4-1}
\usepackage{amsmath,amsfonts,amssymb}
\usepackage{mathtools}
\usepackage{multirow}

\usepackage[utf8]{inputenc}

\usepackage[T1]{fontenc}

\usepackage{bm}
\usepackage{times}

\input{pretex}

\definecolor{beamer}{rgb}{0.2,0.2,0.7}
\definecolor{colorone}{rgb}{1,0.36,0.03}
\definecolor{colortwo}{rgb}{0.4,0.77,0.17}
\definecolor{colorthree}{rgb}{0.01,0.51,0.93}
\definecolor{colorfour}{rgb}{0.47,0.26,0.58}
\definecolor{colorfive}{rgb}{0.12,0.55,0.16}
\usepackage{tcolorbox}
\usepackage{relsize}
\usepackage{graphicx}
\usepackage{booktabs}
\usepackage{amsmath}
\usepackage{tikz}
\usetikzlibrary{quantikz}
\nc{\st}{\text{subject to} \ }
\nc{\supre}{\text{supremum} \ }
\nc{\sdp}{\text{sdp}}
\usepackage{array}

\newcommand{\etal}{\textit{et al.}}
\newcommand{\set}[1]{ \left\{ #1 \right\} }

\allowdisplaybreaks

\newcommand{\trace}[2][]{\tr_{#1}\left[ #2 \right]}
\newcommand{\combwr}[2]{{ \operatorname{Cr}^{#1}_{#2} }}
\DeclarePairedDelimiter{\abs}{\lvert}{\rvert}
\renewcommand{\set}[1]{ \left\{ #1 \right\} }
\newcommand{\setcond}[2]{ \left\{ #1 \mid #2 \right\} }

\nc{\ith}[1]{{#1}^\mathrm{th}}

\begin{document}
\title{Retrieving non-linear features from noisy quantum states}
 
\author{Benchi Zhao}
\affiliation{Thrust of Artificial Intelligence, Information Hub, Hong Kong University of Science and Technology (Guangzhou), Nansha, China}
\affiliation{Graduate School of Engineering Science, Osaka University, 1-3 Machikaneyama, Toyonaka, Osaka 560-8531, Japan}

\author{Mingrui Jing}
\affiliation{Thrust of Artificial Intelligence, Information Hub, Hong Kong University of Science and Technology (Guangzhou), Nansha, China}

\author{Lei Zhang}
\affiliation{Thrust of Artificial Intelligence, Information Hub, Hong Kong University of Science and Technology (Guangzhou), Nansha, China}

\author{Xuanqiang Zhao}
\affiliation{QICI Quantum Information and Computation Initiative, Department of Computer Science, The University of Hong Kong, Pokfulam Road, Hong Kong, China}

\author{Yu-Ao Chen}
\affiliation{Thrust of Artificial Intelligence, Information Hub, Hong Kong University of Science and Technology (Guangzhou), Nansha, China}

\author{Kun Wang}
\affiliation{Thrust of Artificial Intelligence, Information Hub, Hong Kong University of Science and Technology (Guangzhou), Nansha, China}

\author{Xin Wang}
\email{felixxinwang@hkust-gz.edu.cn}
\affiliation{Thrust of Artificial Intelligence, Information Hub, Hong Kong University of Science and Technology (Guangzhou), Nansha, China}

\begin{abstract}
Accurately estimating high-order moments of quantum states is an elementary precondition for many crucial tasks in quantum computing, such as entanglement spectroscopy, entropy estimation, spectrum estimation, and predicting non-linear features from quantum states. But in reality, inevitable quantum noise prevents us from accessing the desired value. In this paper, we address this issue by systematically analyzing the feasibility and efficiency of extracting high-order moments from noisy states. We first show that there exists a quantum protocol capable of accomplishing this task if and only if the underlying noise channel is invertible. We then establish a method for deriving protocols that attain optimal sample complexity using quantum operations and classical post-processing only. Our protocols, in contrast to conventional ones, incur lower overheads and avoid sampling different quantum operations due to a novel technique called observable shift, making the protocols strong candidates for practical usage on current quantum devices. The proposed method also indicates the power of entangled protocols in retrieving high-order information, whereas in the existing methods, entanglement does not help. {We further construct the protocol for large quantum systems to retrieve the depolarizing channels, making the proposed method scalable.} Our work contributes to a deeper understanding of how quantum noise could affect high-order information extraction and provides guidance on how to tackle it.
\end{abstract}

\maketitle

\section{Introduction}
Quantum computing has emerged as a rapidly evolving field with the potential to revolutionize the way we process and analyze information. Such an advanced computational paradigm stores and manipulates information in a quantum state, which forms an elaborate representation of a many-body quantum system~\cite{bennett_quantum_2000}. One critical task for this purpose is to estimate the $k$-th \textit{moment} of a quantum state's density matrix $\rho$, which is often denoted as $\tr[\rho^k], k\in \mathbb{Z}^+$. For example, the second moment of $\rho$ is commonly known as the \textit{purity} of $\rho$. Accurately computing $\tr[\rho^k]$ provides an elementary precondition for extracting spectral information of the quantum state~\cite{nielsen2010quantum,chakraborty2022oneshot},
which is crucial in supporting the evaluation of non-linear functions in quantum algorithms~\cite{childs2017quantum,chen2021variational}, applying to entanglement spectroscopy by determining measures of entanglement, e.g., \textit{R\'enyi entropy} and \textit{von Neumann entropy}~\cite{johri2017entanglement,chung_entanglement_2014}, and characterizing non-linear features of complex quantum systems in materials~\cite{Pollmann2010,Vidal2003,subacsi2019entanglement,Li2008}. In particular, as a core-induced development, understanding and controlling quantum entanglement inspire various quantum information breakthroughs, including entanglement theories, quantum cryptography, teleportation and discrimination~\cite{horodecki2009quantum,yin2020entanglement,pirandola2015advances,hayashi2006bounds}. 

\begin{figure}
    \centering
    \includegraphics[width = \linewidth]{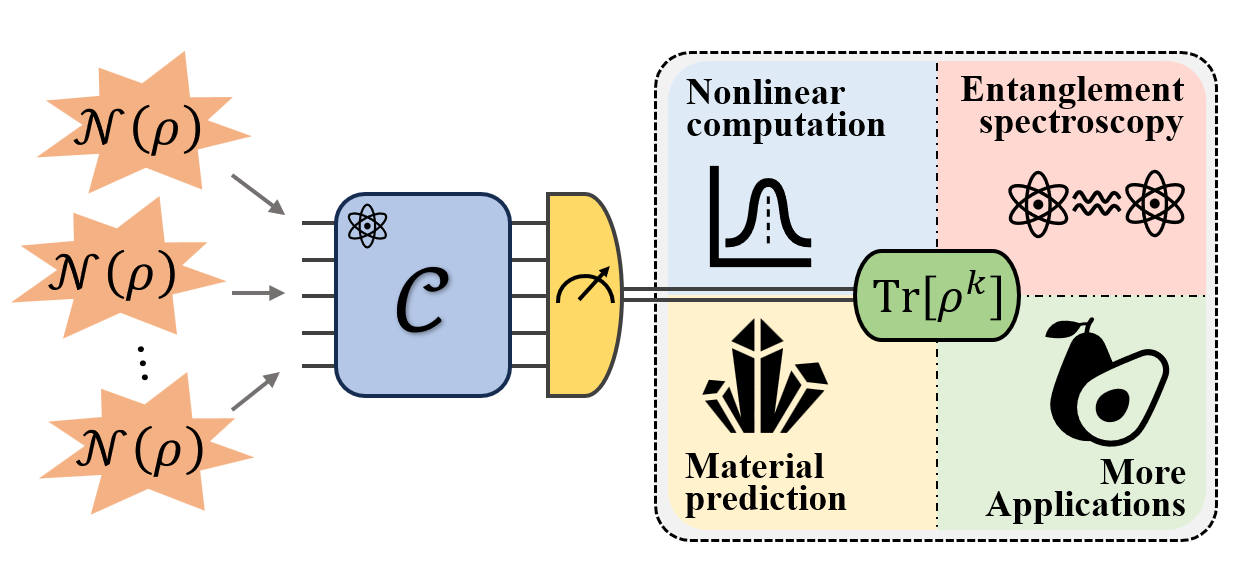}
    \caption{The general framework of recovering the high-order quantum information $\tr[\rho^k]$ given copies of noisy resource $\cN(\rho)$ based on our derived protocol, i.e., a quantum channel $\cC$ and measurement-based post-processing. The information can be further employed in various applications in practical quantum computing.}
    \label{fig:main_fig}
\end{figure}

Numerous methods have been proposed for efficiently estimating quantum state spectra on a quantum computer, including the deterministic quantum schemes processing intrinsic information of the state~\cite{subramanian2021quantum} and the variational quantum circuit learning for approximating non-linear quantum information functions~\cite{Mitarai_2018,tan2021variational}. Meanwhile, a direct estimation method of $\tr[\rho^k]$ through the Newton-Girard method and \textit{generalized swap trick}~\cite{johri2017entanglement} has been proposed in~\cite{subacsi2019entanglement}, and then it was further improved by~\cite{yirka_qubit-efficient_2021}. 

However, quantum systems are inherently prone to the effects of noise, which can arise due to a variety of factors, such as imperfect state preparation, coupling to the environment, and imprecise control of quantum operations~\cite{clerk2010introduction}. In definition, quantum noise can be described in a language of quantum operation denoted as $\cN$. Such an operation can inevitably pose a significant challenge to the reliable estimation of $\tr[\rho^k]$ from corrupted copies of quantum state $\cN(\rho)$.

Previous works concentrated on the first-order situation by applying the inverse operation $\cN^{-1}$~\cite{temme2017error, jiang2021physical, endo2018practical, takagi2021optimal} to each copy of the noisy state, such that $\cN^{-1}\circ\cN={\rm id}$, where {\rm id} means identity map. 
Such inverse operation might not be physically implementable, which requires the usage of the quasi-probability decomposition (QPD) and sampling techniques~\cite{piveteau2022quasiprobability}, decomposing $\cN^{-1}$ into a linear combination of quantum channels $\cN^{-1} = \sum_i c_i \cC_i$. 
Then, the value $\tr[O\rho]$ can be estimated in a statistical manner, and the total required sampling times are square proportional to \textit{sampling overhead} $g = \sum_i|c_i|$~\cite{hoeffding1994probability}. 
Nevertheless, the situations for estimating $\tr[\rho^k]$ with $k>1$ stay unambiguous apart from handling individual state noise. In this paper, we are going to retrieve the $k$-th moment from noisy states, which is illustrated in Fig.~\ref{fig:main_fig}. To systematically analyze the feasibility and efficiency of extracting high-order moment information from noisy states, as shown in Fig.~\ref{fig:main_fig}, The following two questions are addressed:
\begin{enumerate}
    \item \textit{Under what conditions can we retrieve the high-order moments from noisy quantum states?}
    \item \textit{For such conditions, what is the quantum protocol that achieves the optimal sampling complexity?} 
\end{enumerate}
These two questions address the existence and efficiency of quantum protocols for retrieving high-order moment information and essential properties from noisy states, which help us to access accurate non-linear feature estimations.

In the present study, we aim to address both of these questions. For the first question, we establish a necessary and sufficient condition for the retrieval of high-order moments from noisy states, which states that a quantum protocol can achieve this goal if and only if the noisy channel is invertible. Regarding the second question, we propose a quantum protocol that can attain optimal sampling complexity using quantum operations and classical post-processing only. In contrast to the conventional sampling techniques, our protocol only employs one quantum operation due to avoiding quasi-probability decomposition and developing a novel technique called \textit{observable shift}. {We further construct a protocol for large quantum systems to retrieve the depolarizing channels, making the observable shift method scalable.}

We also demonstrate the advantages of our method over existing QPD methods~\cite{temme2017error,endo2018practical,piveteau2022quasiprobability} with step-by-step protocols for some types of noise of common interest. Our protocols incur lower sampling overheads and have simple workflows, serving as strong candidates for practical usage on current quantum devices. The proposed method also indicates the power of entanglement in retrieving high-order information, whereas in the existing methods, entangled protocols do not help~\cite{jiang2021physical, regula2021operational}.
In the end, numerical experiments are performed to demonstrate the effectiveness of our protocol with depolarizing noise applied on the ground state of the Fermi-Hubbard model. Our sampling results illustrate a more accurate estimation on $\tr[\rho^2]$ compared with no protocol applied. 

\section{Moment recoverability}
In this section, we are going to address the first question proposed in the introduction. We discover a necessary and sufficient condition for the existence of a high-order moment extraction protocol as shown in Theorem~\ref{theorem:NS_condition}.

\begin{theorem}\label{theorem:NS_condition}{\rm (Necessary and sufficient condition for existence of protocol)}
    Given a noisy channel $\cN$, there exists a quantum protocol to extract the $k$-th moment $\tr[\rho^k]$ for any state $\rho$ if and only if the noisy channel $\cN$ is invertible.
\end{theorem}

Intuitively, we can understand Theorem~\ref{theorem:NS_condition} from the following aspects.
Estimating high-order moment demands complete information about quantum channels. If a noise channel $\cN$ is invertible, it means information stored in quantum states is deformed, which can be carefully re-deformed back to the original information with extra resources of noisy states and sampling techniques.
However, when the loss of information is unattainable, i.e., the noise is non-invertible. Part of the information stored in the quantum state is destroyed completely, leading to an infeasible estimation problem even with extra quantum resources. An illustration of the theorem is shown in Fig.~\ref{fig:main_coro}.

\begin{figure}
    \centering
    \includegraphics[width=\linewidth]{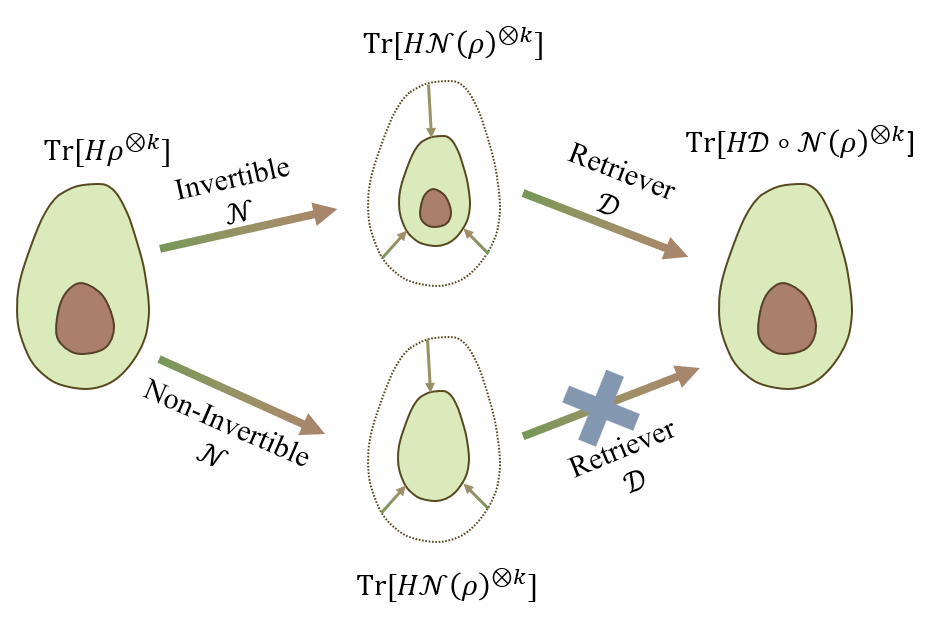}
    \caption{Illustration of Theorem~\ref{theorem:NS_condition}. Suppose a state $\rho$ is corrupted by an invertible channel $\cN$, and $H$ is the moment observable, such that $\tr[H \rho^{\ox k}] = \tr[\rho^k]$. The state information is deformed but can be retrieved via applying $\cD$ and post-processing (Top). However, for non-invertible $\cN$, the high-order moment is completely destroyed and cannot be retrieved (Bottom).}
    \label{fig:main_coro}
\end{figure}

In the following, we will present a sketch of proof for our main theorem. Starting with the definition of a \textit{quantum protocol} which is usually described as a sequence of realizable quantum operations and post-processing steps used to perform a specific task in the domain of quantum information processing. 
Mathematically, we say there exists a quantum protocol to retrieve the $k$-th moment from copies of a noisy state $\cN(\rho)$ if there exists an operation $\cD$ such that 
\begin{equation}\label{eq:protocol}
    \tr[H\cD\circ\cN^{\ox k}(\rho^{\ox k})] = \tr[H\rho^{\ox k}],
\end{equation}
where $H$ is what we call the \textit{moment observable}, as the usage of it is the core of extracting the high-order moment from quantum states, i.e., $\tr[H\rho^{\ox k}] = \tr[\rho^k]$. For example, in estimating the purity of single-qubit states, the moment observable $H$ is just a SWAP operator correlating two qubits. It is proved in the Supplementary Material that for any order $k$, there exists such a moment observable $H$ to extract the $k$-th moment information. The inspiration for our proof comes from the QPD method used to simulate Hermitian-preserving maps on quantum devices, which has enjoyed great success in a variety of tasks, such as error mitigation~\cite{temme2017error, endo2018practical,takagi2021optimal}, and entanglement detection~\cite{peres1996separability, horodecki2009quantum}. We extend our allowed operation $\cD$ to the field covering the Hermitian-preserving maps.

If the noisy channel $\cN$ is invertible, then there exists the inverse operation of the noisy channel $\cN^{-1}$, which is generally a Hermitian-preserving map, and $(\cN^{-1})^{\ox k}$ stands for a feasible solution to the high-order moment retriever.
On the other hand, by assuming a Hermitian-preserving map $\cD$ satisfying Eq.~\eqref{eq:protocol} and non-invertible $\cN$. In the view of the Heisenberg picture, the adjoint of the maps in Eq.~\eqref{eq:protocol} satisfies:
\begin{equation}\label{eq:Heisenberg}
    \tr[(\cN^{\ox k})^\dagger\circ\cD^\dagger (H) \rho^{\ox k}] = \tr[H \rho^{\ox k}].
\end{equation}
It has been proven in \cite{zhao2023information} that given an observable $O$, a Hermitian-preserving map $\cM$ satisfies $\tr[\cM(\rho)O] = \tr[\rho O]$ for any state $\rho$ if and only if it holds that $\cM^\dagger(O)=O$. Thus, we can derive that as long as we find a Hermitian-preserving operation $\cD$ such that the condition 
\begin{equation}\label{eq:condition_for_quantum_channel}
    (\cN^{\ox k})^\dagger\circ\cD^\dagger(H)=H
\end{equation}
is satisfied, the problem is solved. Since the effective rank of $H$ is full, whose definition and proof are shown in Supplementary Material, then from the fact that $\rank(B)\le \min \left(\rank(A), \rank(B)\right)$, we can deduce that $\cN$ is invertible, contradicting to our assumption. This means there exists no quantum protocol for extracting high-order moments when the noise is non-invertible. The detailed proof is given in the Supplementary Material.

\section{Observable shift method}
In the previous part, we mentioned that applying the inverse operation of a noisy channel $\cN^{-1}$ to noisy states simultaneously to mitigate the error is one feasible solution to retrieve high-order moments. However, this channel inverse method requires exponentially many resources with respect to $k$ to retrieve the $k$-th moment. Also, the implementation of inverse operation $\cN^{-1}$ is not quantum device friendly because it has to sample and implement different quantum channels probabilistically.

In this section, we propose a new method called \textit{observable shift} to retrieve high-order moment information from noisy states, which requires only one quantum operation with comparable sampling complexity.
\begin{lemma}\label{lemma:CPTS}{\rm (Observable shift)}
    Given an invertible quantum channel $\cN$ and an observable $O$, there exists a quantum channel $\cC$, called retriever, and  coefficients $t, f$ such that 
    \begin{equation}\label{eq:obs_shift}
    \cN^\dagger\circ\cC^\dagger(O)=\frac{1}{f}\left(O + tI\right),
    \end{equation}
    where $I$ is identity.
\end{lemma}
We develop this observable shift technique since the expectation of $O+tI$ regarding any quantum states can be computed as $\tr[O\rho] + t$ during the measurement procedures. Moreover, if one wants to maintain the retrievability of $\tr[O\rho^k]$ from the noise channel $\cN$ with respect to any possible quantum states, then the only change that could be made to the observable is to add constant identity since such a transformation could maintain the original information of $\rho$. Therefore, the trace value can be retrieved via measurements and post-processing. For instance, when we estimate $\tr[O\rho]$, where $O=I+X+Z$, By skipping the identity, often called \textit{shifting the observable}, the value of $\tr[O\rho]$ can still be re-derived by post-adding a value of one to the expectation value of the \textit{shifted observable} $O' = X+Z$, i.e., $\tr[O\rho] = 1+\tr[O'\rho]$.

Besides, instead of mitigating noise states individually, our method utilizes entanglement to retrieve the information with respect to the moment observable $H$. Compared with the channel inverse method, the proposed observable shift method requires fewer quantum resources, and its implementation is easier. The proposed observable shift method leads to Proposition~\ref{prop:sample_complexity}.

\begin{proposition}\label{prop:sample_complexity}
    Given error tolerance $\delta$, the $k$-th moment information can be retrieved by sampling one quantum channel with complexity $\mathcal{O}(f_{\min}^2(\cN,k)/\delta^2)$ and post-processing. The quantity $f_{\min}(\cN,k)$ is the sampling overhead defined as
\begin{align}\label{eq:f(N,k)}
    f_{\min}(\cN, k) = \min\Big\{f \mid &(\cN^{\ox k})^\dagger\circ \cC^\dagger(H) = \frac{1}{f}\left(H+tI\right),\nonumber\\
    &f\in\mathbb{R}^+, t\in\mathbb{R}, \cC\in{\rm CPTP} \Big\},
\end{align}
where $\cN$ is the noisy channel, $\cC$ is quantum channel, $t$ is the shifted distance, $H$ is the moment observable.
\end{proposition}

In our cases, we aim to find a Hermitian-preserving map $\cD$ such that Eq.~\eqref{eq:condition_for_quantum_channel} holds together with the allowance of observable shifting~\eqref{eq:obs_shift}, i.e.,
\begin{equation}
    (\cN^{\ox k})^\dagger\circ\cD^\dagger(H)=H' - tI,
\end{equation}
where $H' = H+tI$ is the shifted observable, and $t$ is a real coefficient. Note that the quantum channel $\cN^{\ox k}$ is a completely positive and trace preserving (CPTP) map, and the adjoint of a CPTP map is completely positive unital preserving~\cite{khatri2020principles}, which refers $(\cN^{\ox k})^\dagger(I)=I$. Thus we have 
\begin{equation}
    (\cN^{\ox k})^\dagger(\cD^\dagger(H)+tI)=H',
\end{equation}
where we can consider $\cD^\dagger(H)+tI$ as a whole and denote it as $\Tilde{\cD}(H)$.
With proper coefficient $t$, the map $\Tilde{\cD}$ could reduce to a completely positive map $\cC$. The detailed proof is given in the Supplementary Material.

If we apply the quantum channel $\cC$ to a noisy state and make measurements over moment observable $H$, the expectation value will be 
\begin{align}
    \zeta &= \tr[H\cC\circ\cN^{\ox k}(\rho^{\ox k})] =\tr[(\cN^{\ox k})^\dagger\circ\cC^\dagger(H)\rho^{\ox k}]\\
    &= \frac{1}{f}\tr[(H + tI)\rho^{\ox k}] = \frac{1}{f}(\tr[H\rho^{\ox k}] +t).
\end{align}
Obviously, the desired high-order moment is given by $\tr[H\rho^{\ox k}]  = f\zeta - t$. In order to obtain the target expectation value of $\tr[H\rho^{\ox k}]$ within an error $\delta$ with a probability no less than $1-p$, the number of total sampling times $T$ is given by Hoeffding's inequality~\cite{hoeffding1994probability},
\begin{equation}\label{eq:hoeffding}
    T \geq f^2 \frac{2}{\delta^2}\log(\frac{2}{p}).
\end{equation}

Usually, the success probability $1-p$ is fixed. Thus we consider it as a constant in this paper, and the corresponding sample complexity is $\cO(f^2/\delta^2)$, which only depends on error tolerance $\delta$ and the \textit{sampling overhead} $f$. It is desirable to find a quantum retriever $\cC$ and shift distance $t$ to make the sampling overhead $f$ as small as possible.
The optimal sampling overhead $f_{\min}(\cN, k)$ of our method can be calculated by SDP as follows:

\begin{subequations}\label{eq:total_SDP_for_obs_shift}
\begin{align}
    f_{\min}(\cN, k) &= \min  \quad f\\
    \st &\quad J_{\Tilde{\cC}} \ge 0 \label{eq:total_SDP_obs_CP}\\
    &\quad \tr_{C}[J_{\Tilde{\cC}_{BC}}]=f I_{B}\label{eq:total_SDP_obs_TS}\\
    &\quad J_{\cF_{AC}} \equiv \tr_B[(J_{\cN_{AB}^{\ox k}}^{T_B}\ox I_C) (I_A\ox J_{\Tilde{\cC}_{BC}})]\label{eq:total_SDP_obs_compose}\\
    &\quad \tr_C[(I_A\ox H_C^T)J_{\cF_{AC}}^T] = H_A + tI_A.\label{eq:total_SDP_obs_constraint}
\end{align}
\end{subequations}
The $J_{\Tilde{\cC}}$ and $J_{\cN^{\ox k}}$ are the \Choi matrices for the completely positive trace-scaling map $\Tilde{\cC} =f\cC$ and noise channel $\cN^{\ox k}$ respectively. Eq.~\eqref{eq:total_SDP_obs_CP} corresponds to the condition that the map $\Tilde{\cC}$ is completely positive, and Eq.~\eqref{eq:total_SDP_obs_TS} guarantees that $\Tilde{\cC}$ is a trace-scaling map. In Eq.~\eqref{eq:total_SDP_obs_compose}, $J_\cF$ is the Choi matrix of the composed map $\Tilde{\cC} \circ \cN^{\ox k}$. Eq.~\eqref{eq:total_SDP_obs_constraint} corresponds to the constraint shown in Eq.~\eqref{eq:obs_shift}.

Beyond retrieving particular non-linear features, i.e., $\tr[\rho^k]$, we can also apply our method to estimate non-linear functions. For a toy example, if we wish to estimate the function $F(\rho)=\frac{1}{2}\tr[\rho^2]+\frac{1}{3}\tr[\rho^3]$, we should design the moment observable first, which is supposed to be $H = \frac{1}{2}H_2\ox I + \frac{1}{3}H_3$, where $H_2$ and $H_3$ are the moment observables for two and three qubits respectively. Then, the retrieving protocol with optimal sampling overhead is given by SDP as shown in Eq.~\eqref{eq:total_SDP_for_obs_shift}. {With the power of estimating non-linear functions, our method applies to the entropy evaluation from noisy quantum states, showcasing the practical potential in quantum many-body correlations determination and entanglement detection. One can directly assist  the estimation of \textit{R\'enyi entropy}~\cite{muller2013quantum} provided a quantum state $\rho$, as,
\begin{equation}
    H_{\alpha}(\rho):=\frac{1}{1-\alpha}\log\left(\tr[\rho^{\alpha
    }]\right),
\end{equation}
where $\alpha\in(0,1)\cup(1,\infty)$. 
}

\section{Protocols for particular noise channels}

We have introduced the observable shift method in the previous part, next will provide the analytical protocol for retrieving the second-order moment information $\tr[\rho^2]$ from noisy quantum states suffering from depolarizing channel and amplitude dimpling channel, respectively.

Depolarizing channels have been extensively studied due to its simplicity and ability to represent a wide range of physical processes that can affect quantum states~\cite{nielsen2010quantum}. A quantum state that undergoes a depolarizing channel would be randomly replaced by a maximally mixed state with a certain error rate.
The single-qubit depolarizing (DE) noise $\cN_{\rm DE}^\epsilon$ has an exact form,
\begin{equation}
    \cN_{\rm DE}^\epsilon(\rho)=(1-\epsilon)\rho + \epsilon \frac{I}{2},
\end{equation}
where $\epsilon$ is the noise level, and $I$ refers to the identity operator.

Given many copies of such noisy quantum states, our method derives a protocol for retrieving the second-order moment $\tr[\rho^2]$ using only one quantum channel and post-processing. Specifically, we have Proposition~\ref{prop:retrieve_DE}.
\begin{proposition}\label{prop:retrieve_DE}
    Given two copies of noisy states, $\cN^\epsilon_{\rm DE}(\rho)^{\ox 2}$, and error tolerance $\delta$, the second order moment $\tr[\rho^2]$ can be estimated by $f\trace{H\cC\left(\cN^\epsilon_{\rm DE}(\rho)^{\ox 2}\right)} - t$, with optimal sample complexity $\mathcal{O}(1/(\delta^2 (1-\epsilon)^4)$, where $f=\frac{1}{(1-\epsilon)^2}$, $t = \frac{1-(1-\epsilon)^2}{2(1-\epsilon)^2}$. The term $\trace{H\cC\left(\cN^\epsilon_{\rm DE}(\rho)^{\ox 2}\right)}$ can be estimated by implementing a quantum retriever $\cC$ on noisy states and making measurements over moment observable $H$. Moreover, there exists an ensemble of unitary operations $\{p_j, U_j\}_j$ such that the action of the retriever $\cC$ can be interpreted as,
    \begin{equation}
        \cC(\cdot) = \sum_{j=1} p_j 
        U_j (\cdot) U_j^\dagger
    \end{equation}
\end{proposition}
We derive an explicit form of a mixed-unitary ensemble containing twelve fixed unitary $U_j$'s given in the Supplementary Material. 

As a result, the second order moment can be retrieved from depolarized states $\cN^\epsilon_{\rm DE}(\rho)^{\ox 2}$ by applying the unitaries $U_j$ randomly with equal probabilities and then performing measurements with respect to the moment observable $H$. After repeating these steps for $T$ rounds, where $T$ is given by Eq.~\eqref{eq:hoeffding}, and averaging the measurement results, we can obtain the estimated expectation value $\zeta = \tr[H\cC(\cN^\epsilon_{\rm DE}(\rho)^{\ox 2})]$. Then, the desired second-order moment is given by 
\begin{equation}
     \tr[\rho^2]= \frac{1}{(1-\epsilon)^2} \zeta - \frac{1-(1-\epsilon)^2}{2(1-\epsilon)^2}.
\end{equation}

When estimating $\tr[\rho^2]$ from copies of the noisy state $\cN^\epsilon_{\rm DE}(\rho)^{\ox 2}$, QPD-based methods incurs a sampling overhead $\frac{(1+\epsilon/2)^2}{(1-\epsilon)^2}$. On the other hand, our observable shift method offers a protocol with a lower sampling overhead $\frac{1}{(1-\epsilon)^2}$, which is much lower than that of the QPD-based method.

Besides, the quantum amplitude damping (AD) channel is another important model that we are interested in, which often appears in superconducting qubits or trapped ions. This type of noise is particularly relevant for the loss of energy or the dissipation of excited states~\cite{breuer2002theory}, whose action results in the transition of a qubit's excited state to its ground state, offering a more realistic representation of energy relaxation processes in quantum systems. The AD channel $\cN_{\rm AD}^\varepsilon$ is characterized by a single parameter $\varepsilon$, representing the damping rate, which has two Kraus operators: $A_0^\varepsilon \coloneqq \proj{0} + \sqrt{1-\varepsilon}\proj{1}$ and $A_1^\varepsilon \coloneqq \sqrt{\varepsilon}\ketbra{0}{1}$, where $\varepsilon\in[0,1]$. Similarly, given many copies of AD-produced quantum states, the second-order information $\tr[\rho^2]$ can be retrieved by applying only one quantum channel and post-processing with our protocol in Proposition~\ref{prop:retrieve_AD}.
\begin{proposition}\label{prop:retrieve_AD}
    Given two copies of noisy states, $\cN^\varepsilon_{\rm AD} (\rho)^{\ox 2}$, and error tolerance $\delta$, the second order moment $\tr[\rho^2]$ can be estimated by $f\trace{H\cC\left(\cN^{\varepsilon}_{\rm AD}(\rho)^{\ox 2}\right)} - t$, with optimal sample complexity $\mathcal{O}(1/(\delta^2 (1-\varepsilon)^4)$, where $f=\frac{1}{(1-\varepsilon)^2}$, $t = -\frac{\varepsilon^2}{(1-\varepsilon)^2}$. The term $\trace{H\cC\left(\cN^{\varepsilon}_{\rm AD}(\rho)^{\ox 2}\right)}$ can be estimated by implementing a quantum retriever $\cC$ on noisy states and making measurements. Moreover, the Choi matrix of such the retriever $\cC$ is
    \begin{align}\label{eq:AD_choi}
        J_\cC &= \proj{00} \ox \frac{1}{6} ((1+2\varepsilon)II + (1-4\varepsilon)H)\nonumber\\
        &\quad+ \proj{\Psi^+} \ox \frac{1}{6}((1+2\varepsilon)II + (1-4\varepsilon)H) \nonumber\\
        &\quad + \proj{\Psi^-} \ox \frac{1}{2} (II - H) \nonumber\\
        &\quad+ \proj{11} \ox \frac{1}{6}(II + H),
    \end{align}
where $\proj{\Psi^\pm}=\frac{1}{2}(\ket{01}\pm\ket{10})(\bra{01}\pm\bra{10})$ are Bell states.
\end{proposition}
The above retriever $\cC$ can be implemented based on the following measurements and post-processing. Given amplitude damping noisy states $\cN^\varepsilon_{\rm AD}(\rho)^{\ox 2}$, we make measurements in the basis $\cB = \{\ket{00}, \ket{\Psi^+}, \ket{\Psi^-}, \ket{11}\}$.
From the Choi matrix of the retriever $\cC$, which is shown in Eq.~\eqref{eq:AD_choi}, we know that based on the obtained measurement results, the quantum system collapses to the states
$\sigma_1, \sigma_2, \sigma_3,\sigma_4$ correspondingly, where $\sigma_1 = \sigma_2 = \frac{1}{6}((1+2\varepsilon)II + (1-4\varepsilon)H)$,  $\sigma_3 = \frac{1}{2}(II-H)$ and $\sigma_4 = \frac{1}{6}(II+H)$. Each state corresponds to a fixed expectation value $\tr[H\sigma_i]$, which can be predetermined via direct matrix calculation with fixed $H$ and known  $\sigma_i$.
The next step is to run sufficient shots of basis-$\cB$ measurements to determine the probability of measuring each basis state, denoted as $p_i$, respectively. The term $\trace{H\cC\left(\cN^{\varepsilon}_{\rm AD}(\rho)^{\ox 2}\right)}$ is then given by the estimated value $\zeta = \sum_{i=1}^4 p_i\tr[H \sigma_i]$. The desired second-order moment is obtained by
\begin{equation}
    \tr[\rho^2] = \frac{1}{(1-\varepsilon)^2}\zeta + \frac{\varepsilon^2}{(1-\varepsilon)^2}.
\end{equation}
More details can be found in the Supplementary Material. The sampling overhead for QPD-based methods is $\frac{(1+\varepsilon)^2}{(1-\varepsilon)^2}$, while the overhead incurred by our method is still as low as $\frac{1}{(1-\varepsilon)^2}$, saying that our method requires fewer quantum resources.

\section{Generalized observable shift method}
{The specific protocol of the observable shift method with minimal sampling overhead is given by SDP as shown in Eq.~\eqref{eq:total_SDP_for_obs_shift}. However, when the size of the system increases, the computer memory for solving such SDP increases exponentially. In personal computers, a 5-qubit system is the largest problem size that SDP can solve. Achieving a specific error mitigation protocol on a large system is a crucial problem.}

{We present a scalable approach for creating an error mitigation protocol by exploiting the observable shift method recursively. More specifically, we demonstrate the capability to design a quantum protocol for mitigating depolarizing noises on arbitrary copies of noisy qudits, as evidenced in the subsequent Proposition~\ref{prop:kth moment n qudit}.}

\begin{proposition}~\label{prop:kth moment n qudit}
    {Given arbitrary $k$ copies of noisy states $\cN^\epsilon_{\rm DE}(\rho)^{\ox k}$, the $k$-th order moment $\tr[ \rho^k ]$ can be estimated by $f_k \trace{H \cC_k \left(\cN^\epsilon_{\rm DE}(\rho)^{\ox k}\right)} - t_k$, where the term $\trace{H \cC_k \left(\cN^\epsilon_{\rm DE}(\rho)^{\ox k}\right)}$ can be estimated by implementing a quantum retriever $\cC_k$ on noisy states and making measurements. Moreover, such $\cC_k$, $f_k$, $t_k$ can be recursively constructed as
\begin{align}
    f_k &= \frac{1}{(1 - \eps)^k},\\
    \cC_k &= {\rm id}_k + \sum_{l = 2}^{k  - 1} \binom{k}{l} (1 - \eps)^{l}\eps^{k - l} f_l \cR_{l}^\dag \circ \left(\cC_{l} \ox {\rm id}_{k - l} \right), \text{ and }\\
    t_k &= f_k \left[ \frac{\eps^k}{d^k} + \frac{k (1 - \eps) \eps^{k - 1}}{d^{k - 1}} 
    - \sum_{l = 2}^{k  - 1} \binom{k}{l} \frac{ (1 - \eps)^{l}\eps^{k - l} }{d^{k - l}} t_l \right]
,\end{align}
    for $\cC_2, t_2$ given in Proposition~\ref{prop:retrieve_DE} and some CP maps $\cR_l$.}
\end{proposition}

{Note that the implementation of $\cC_k$ requires post-selection of measurement outcomes, as $\cC_k$ is completely positive but not necessarily trace-preserving. The detailed proof of Proposition~\ref{prop:kth moment n qudit} are deferred to Appendix~\ref{appendix:depo n qudit k moment}. We also numerically verify the feasibility of the number of state copies to be up to a hundred.}

\section{Comparison with existing protocols}
To extract high-order moment information from noisy states, one straightforward method is to apply an inverse operation of noisy channel $\cN^{-1}$ on quantum states to mitigate error, and then perform measurements over moment observable $H$, which is $\tr[H\rho^{\ox k}] = \tr[H \left(\cN^{-1} \right)^{\ox k} \left(\cN (\rho)^{\ox k} \right)]$. However, the map $\cN^{-1}$ might not be a physical quantum channel~\cite{jiang2021physical}, thus we cannot implement it directly on a quantum system. Fortunately, we can simulate such channel by quasi-probability decomposition, which decomposes such non-physical map into a linear combination of physical quantum channels, i.e., $\cN^{-1} = \sum_i c_i \cC_i$, where $c_i$ are the real coefficients and $\cC_i$ are physical quantum channels.  We need to note that $c_i$ can be negative. From the aspect of physical implementation, in the $t$-th round of total $T$ times of sampling, we first sample a quantum channel $\cC^{(t)}$ from $\{\cC_i\}$ with probability $\{|c_i|/g\}$, where $g = \sum_i |c_i|$, and apply it to noisy state $\cC^{(t)}\circ\cN(\rho)$. Then we take measurements and get results $o^{(t)}$. After $T$ rounds of sampling, we attain an estimation for the expectation value $\zeta = \frac{g}{T}\sum_{t=1}^T \text{\rm sgn}(c_i^{(t)})o^{(t)} = \tr[O\rho]$. The total sampling times $T$ is also given by Hoeffding's inequality as shown in Eq.~\eqref{eq:hoeffding}. The optimal sampling overhead $g_{\min}(\cN)$ is given by \cite{jiang2021physical, zhao2023information, regula2021operational}
\begin{align}
    g_{\min}(\cN) = \min\Big\{\sum_i |c_i| \mid& \cN^{-1} = \sum_i c_i \cC_i, \nonumber \\ 
    & c_i\in\mathbb{R}, \cC_i\in \text{CPTP}\Big\},
\end{align}
which can be obtained by SDP as displayed in the Supplementary Material.

When we apply the channel inverse method to retrieve the $k$-th moment, we should apply the inverse operation simultaneously on $k$ quantum systems, the corresponding optimal sampling overhead if given by $g_{\min}(\cN, k)$, which is 
\begin{align}\label{eq:g(N,k)}
    g_{\min}(\cN, k) = \min\Big\{\sum_i |c_i| \mid& (\cN^{-1})^{\ox k} = \sum_i c_i \cC_i, \nonumber \\
    & c_i\in\mathbb{R}, \cC_i\in \text{CPTP}\Big\},
\end{align}

\begin{figure}[ht]
\centering
\includegraphics[width=0.9\linewidth]{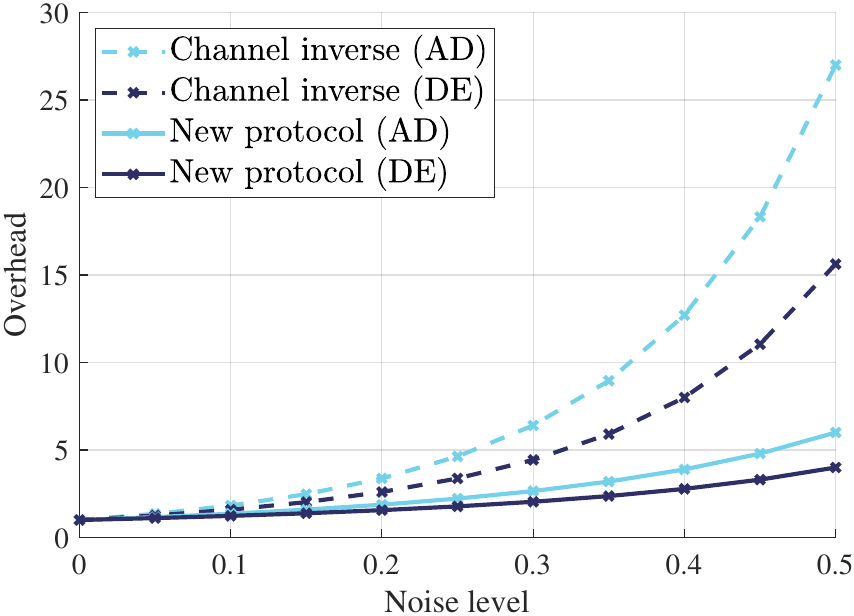}
    \caption{The sampling overhead respective to noise level for estimating $\tr[\rho^3]$ from amplitude damping noise corrupted state $\cN(\rho)$. The dashed curve refers to the overhead for the channel inverse, and the solid curve stands for the newly proposed method. The light blue and dark blue curves represent results from amplitude damping (AD) and depolarizing channel (DE), respectively.}
    \label{fig:overhead_AD_k=3}
\end{figure}

With Eq.\eqref{eq:f(N,k)} and Eq.~\eqref{eq:g(N,k)}, we can make a comparison of the sampling overhead between our method and the conventional QPD channel inverse method, which leads to Lemma~\ref{lemma:advantage}. 

\begin{lemma}\label{lemma:advantage}
    For arbitrary invertible quantum noisy channel $\cN$, and moment order $k$, we have 
    \begin{equation}
        f_{\min}(\cN, k) \le g_{\min}(\cN, k)
    \end{equation}
\end{lemma}

Lemma~\ref{lemma:advantage} implies that in the task of extracting the $k$-th moment $\tr[\rho^k]$ from noisy states, the proposed method requires fewer sampling times, i.e., consumes fewer quantum resources. The detailed proof is displayed in Supplementary Material. It has been proven by Regula \etal in Ref.~\cite{regula2021operational} that the sampling overhead for simulating a trace preserving linear map is equivalent to its diamond norm. Specifically, in the case of inverse operation $\cN^{-1}$, we have $g_{\min}(\cN) = \|\cN^{-1}\|_\diamond$. Note that the diamond norm is multiplicativity with respect to tensor product~\cite{regula2021operational, jiang2021physical}, i.e., $\|(\cN^{-1})^{\ox k}\|_\diamond = \|\cN^{-1}\|_\diamond^k$. Thus, we conclude that the optimal sampling overhead for retrieving high-moment from noisy states increases exponential respect to the moment order $k$, which is 
\begin{align}
    g_{\min}(\cN, k) &= g_{\min}(\cN^{\ox k}) = \|(\cN^{-1})^{\ox k}\|_\diamond
    = \|\cN^{-1}\|_\diamond^{k} \nonumber\\
    &= g_{\min}(\cN)^k.
\end{align}

In order to illustrate the advantage of the proposed method over the channel inverse method in terms of sampling overhead, we conduct a numerical experiment to extract the third moment $\tr[\rho^3]$ from amplitude damping noise channel with different noise levels. The results are shown in Fig.~\ref{fig:overhead_AD_k=3}. The dashed and solid curves stand for the sampling overhead for the channel inverse and observable shift method, respectively. {We also made a comparison with another QPD-based method which was proposed in~\cite{zhao2023information}. The details are shown in the Appendix~\ref{appen:comparison_with_diff_methods}.}

Compared with the QPD-based inverse operation method, the proposed observable shift method has at least two-fold advantages. First, our method can achieve a lower sampling overhead, i.e., $f_{\min}(\cN, k) \le g_{\min}(\cN, k)$. In previous, we have shown examples where $f_{\min}(\cN, k)$ is strictly smaller, indicating the effectiveness of the newly proposed observable shift technique.
It is also observed that the optimal protocol given by our method is generally an entangled one. In contrast to the uselessness of entanglement in QPD~\cite{jiang2021physical, regula2021operational}, our method demonstrates the power of entanglement for tackling noise.
Second, protocols given by our method are more hardware friendly as they only need to repeat one fixed quantum channel, whereas the QPD-based methods have to sample from multiple quantum channels and implement each of them.
 
\section{Application to Fermi-Hubbard model}

The Fermi-Hubbard model is a key focus in condensed matter physics due to its relevance in metal-insulator transitions and high-temperature superconductivity~\cite{Cade_2020,Dagotto_1994}. Recent studies have shown that entanglement spectroscopy can be utilized to extract critical exponents and phase transitions in the Fermi-Hubbard model~\cite{Kokail_2021,linke2018measuring,szasz2020chiral}. As the model is characterized by a broad range of correlated electrons, it necessitates multi-determinant and highly accurate calculations~\cite{ferreira2022quantum,szasz2020chiral} which hence demand ingenious methods of quantum noise control.

In a physical system such as a metallic crystal with an $n_x\times n_y$ square lattice, each lattice point, known as a site, is assigned an index. The Hubbard model Hamiltonian takes on a fermionic form in second quantization,
\begin{equation}\label{eq:fermihubbard}
\begin{aligned}
    H_{\text{Hubbard}} &= -J\sum_{\langle i,j\rangle,\sigma}(a^\dag_{i\sigma}a_{j\sigma}+a^\dag_{j\sigma}a_{i\sigma})\\ &+ U\sum_i n_{i\uparrow}n_{i\downarrow} + H_{\text{local}},
\end{aligned}
\end{equation}
where $a^\dag_{i\sigma},a_{i\sigma}$ are fermionic creation and annihilation operators; $n_{i\sigma} = a^\dag_{i\sigma},a_{i\sigma}$, are the number operators; the notation $\langle i,j\rangle$ associates adjacent sites in the $n_x\times n_y$ rectangular lattice; $\sigma\in\{\uparrow,\downarrow\}$ labels the spin orbital. The first term in Eq.~\eqref{eq:fermihubbard} corresponds to the hopping term, where $J$ denotes the tunneling amplitude. The second term involves the on-site Coulomb repulsion, represented by $U$. The final term in the equation defines the local potential resulting from nuclear-electron interaction, which we have chosen to be the Gaussian form~\cite{Wecker_2015}.
\begin{equation}
    H_{\text{local}} = \sum_{j=1} \sum_{\nu = \uparrow,\downarrow} \epsilon_{j,\nu} n_{j, \nu}; \quad \epsilon_{j,\nu} = -\lambda_{\nu} e^{-\frac{1}{2}(j-m_{\nu})^2 / \sigma_\nu^2}.
\end{equation}

In the following, we consider a specific 3-site (6-qubit) Fermi-Hubbard Hamiltonian with $J=2,U=3$ and $\lambda_{\uparrow, \downarrow}=3, 0.1$, $m_{\uparrow, \downarrow} = 3, 3$. The standard deviation $\sigma_v$ for both spin-up and -down potentials are set to $1$ guaranteeing a charge-spin symmetry around the center site ($i=2$) of the chain system.

The ground state entanglement spectroscopy of the model identifies the topological-ordering signatures of the system which requires high-precision entropy estimations over each bipartite sector of the entire system. We show that the mitigation of the quantum noise can be achieved and therefore, enhance the determination of $\tr[\cN(\rho_A)^2]$, via our proposed method, which is displayed in the Supplementary Material.

Fig.~\ref{fig:simulation} displays the sampling distribution with and without error mitigation. The orange curve refers to the estimation distribution of second-order information $\tr[\rho^2]$ from noisy states, and the cyan curve shows the estimation distribution with error mitigation. And the black dash line is the exact value of $\tr[\rho^2]$.
\begin{figure}[]
\centering
\includegraphics[width=\linewidth]{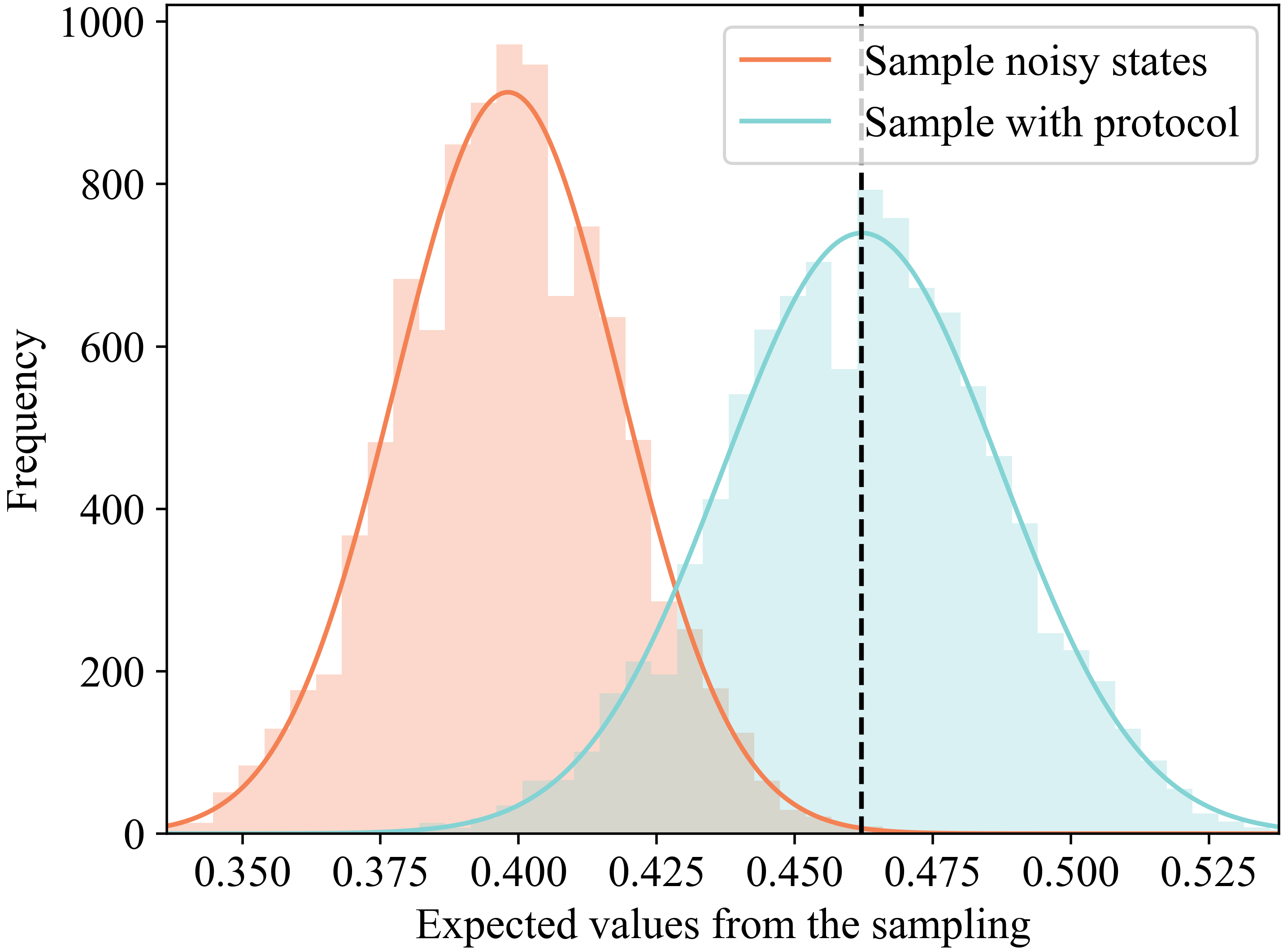}
    \caption{Simulation of high-order moment $\tr[\rho^2]$ estimation. The curves are calculated from sampling. The orange curve represents the estimation from depolarizing noised state $\cN(\rho)$ with noise level $\epsilon=0.1$. The cyan curve is the estimation with the proposed error mitigation method. The black dashed line stands for the exact value of $\tr[\rho^2]$.}
    \label{fig:simulation}
\end{figure}

\vspace{5mm}
\section{Conclusion and discussion}
In this study, we establish that when quantum states are distorted by noises, the original moment information can still be retrieved through post-processing if and only if the noise is invertible. 
Furthermore, our proposed method, called observable shift, outperforms QPD-based techniques in two aspects: (1) The proposed method requires lower quantum sampling complexity than the existing one, which implies the superiority of entangled protocols over product protocols. This contrasts with the multiplicativity of overhead observed in QPD-based methods for quantum error mitigation. (2) The observable shift method is easier to implement than the QPD-based method as it only involves a single quantum operation, which makes our method more friendly to quantum devices. {We also propose the construction of a protocol to retrieve the depolarizing channel of large-size quantum systems.} Our findings have implications for the dependable estimation of non-linear information in quantum systems and can influence various applications, including entanglement spectroscopy and ground-state property estimation. 


For further work, one important task is to improve the scalability of the observable shift method, which makes this approach more practical. Also, investigating the approximate version of retrieving non-linear features is also interesting. We expect the observable shift technique can be incorporated into more algorithms and protocols to boost efficiency. {It will be also interesting to explore other error mitigation methods~\cite{Cai2022,Endo2018a,Cai2020,Koczor2020,Cai2021,Huggins2021} to extract non-linear features from noisy quantum states.}

\textbf{Code availability.--}
In our numerical experiments, we computed SDPs by the package CVX~\cite{cvx, gb08}. The code for numerical simulation is operated on Paddle Quantum~\cite{paddlequantumn}. The code has been uploaded to Github, which can be found at {\url{https://github.com/Dragon-John/high-moment-info}}.

\textbf{Acknowledgements.--}
The authors would like to thank Chengkai Zhu and Chenghong Zhu for their valuable discussion. B.Z would like to thank Giulio Chiribella for helpful his comments. This work was supported by the Start-up Fund (No. G0101000151) from The Hong Kong University of Science and Technology (Guangzhou), the Guangdong Provincial Quantum Science Strategic Initiative (No. GDZX2303007), and the Education Bureau of Guangzhou Municipality.

\bibliography{references}

\clearpage

\vspace{2cm}
\onecolumngrid
\vspace{2cm}
\begin{center}
{\textbf{\large Appendix for \\Retrieving non-linear features from noisy quantum states}}
\end{center}

\appendix

\section{Symbols and notations} \label{sec:symbols_and_notations}
In this work, we focus on linear maps $\cN$ having the same input and output dimension. A linear map $\cN$ is called Hermitian-preserving (HP) if, for any introduced reference system $R$, the product map $id_R\ox \cN$ maps any Hermitian operator to another Hermitian operator. If additionally the map also preserves the positivity of the operators. We say it is completely positive (CP). Besides, a linear map $\cN$ is trace scaling (TS) means $\tr[\cN(X)] = a\tr[X]$ for all $X\in \cL_A, a\in\mathbb{R}$. If $a=1$, such map $\cN$ is called trace-preserving (TP). We call a linear map $\cN_{A\rightarrow A'}$ a quantum channel if it is both \textit{completely positive} and \textit{trace preserving} (CPTP). More, by saying $\cN$ is unital scaling (US), we mean $\tr[\cN(I)] = b\tr[I]$, where $I$ refers to the identity operator, and $b\in\mathbb{R}$. Again, if $b=1$, we say $\cN$ is unital preserving (UP). A sketch of the relationships among these classes of linear maps has been shown in the Venn diagram~\ref{fig:Venn}.

\noindent\textbf{Remark:} Adjoint map has the following properties~\cite{khatri2020principles,zhao2023information} 
\begin{itemize}
    \item The adjoint of a CP map is CP;
    \item The adjoint of a TP map is UP;
    \item The adjoint of a UP map is TP;
    \item The adjoint of a TS map is US;
    \item The adjoint of a US map is TS;
\end{itemize}

\begin{figure}[hptb]
    \centering
    \includegraphics[width = 0.6\linewidth]{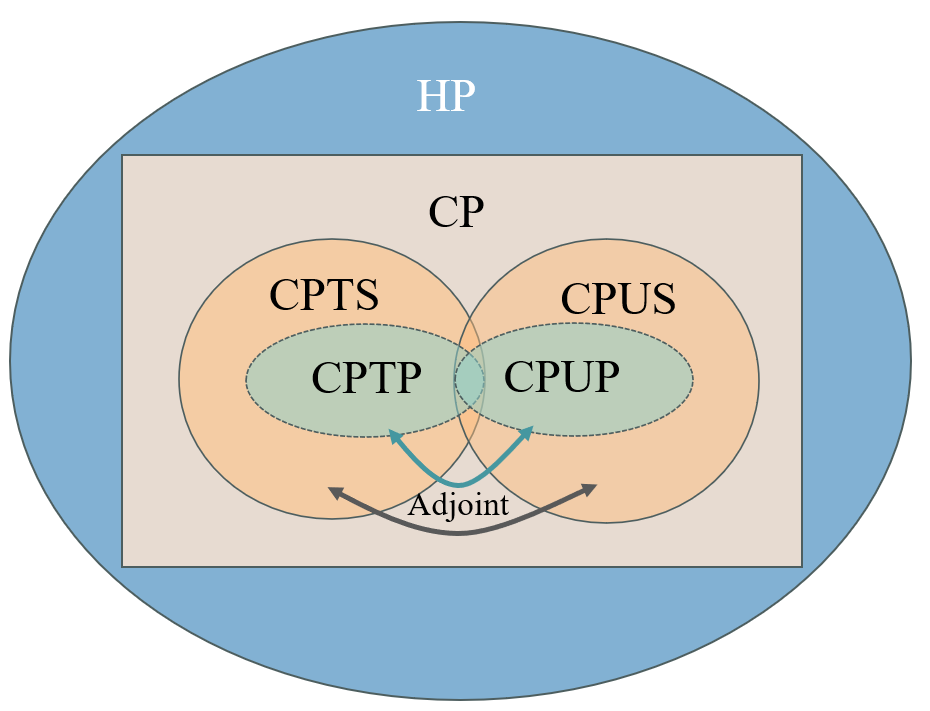}
    \caption{The Venn diagram summarizing relationships among different classes of linear maps stated in the section of notations.}
    \label{fig:Venn}
\end{figure}
\clearpage

\section{Proof of the existence of observable for high-order moment} 
\begin{lemma}~\label{lemma:exitance_of_H}
    Suppose $\rho$ is an $n$-qubit state. Then for a positive integer $k$, there exists an $nk$-qubit Hermitian matrix $H$ such that $\trace{ H \rho^{\ox k} } = \trace{\rho^k}$.
\end{lemma}

\begin{proof}
    Consider the computational basis $\{ \ket{x} \}_x$ for an $n$-qubit system. The desired Hermitian matrix can be constructed from an $nk$-qubit cyclic permutation matrix
\begin{equation} \label{eq:s_k decomp}
    S_k = \sum_{\bm x \coloneqq x_1 \cdots x_k = 0}^{2^{nk} - 1} \bigotimes_{j = 1}^k \ketbra{x_{\pi(j)}}{x_j}
,\end{equation}
    where $\pi = (1, \ldots, k)$ is a permutation function in cyclic notation. Note that the conjugate transpose of this matrix permutes the subsystems in an inverse order i.e.
\begin{equation} \label{eq:s_k^dag decomp}
    S_k^\dagger = \sum_{\bm x} \bigotimes_{j = 1}^k \ketbra{x_j}{x_{\pi(j)}} = \sum_{\bm x} \bigotimes_{j = 1}^k \ketbra{x_{\pi^{-1}(j)}}{x_j}
.\end{equation}
    In the rest of this proof, we will prove that 
\begin{equation}
    \trace{ S_k \rho^{\ox k} } = \trace{ S_k^\dag \rho^{\ox k} } = \trace{ \rho^k }
\end{equation}
    for all $\rho$, and hence the result follows by setting $H = \frac{1}{2} \left( S_k + S_k^\dag \right)$.

    To start with, decomposing $\rho = \sum_{x, y} C_{xy} \ketbra{x}{y}$ gives
\begin{align}
    \rho^k &= \sum_{\bm x, \bm y} \prod_{j = 1}^k C_{x_j y_j} \ketbra{x_j}{y_j} = \sum_{\bm x, y_k} \left( \prod_{j = 1}^{k - 1} C_{x_j x_{\pi(j)}} \right) c_{x_k y_k} \ketbra{x_1}{y_k} \implies \trace{ \rho^k } = \sum_{\bm x} \prod_{j = 1}^{k} C_{x_j x_{\pi(j)}}, \\
    \rho^{\ox k} &= \sum_{\bm x, \bm y}  \bigotimes_{j = 1}^k C_{x_j y_j} \ketbra{x_j}{y_j}
.\end{align}
    Apply $S_k$ on $\rho^{\ox k}$ and we have
\begin{align}
    S_k \rho^{\ox k} &= \sum_{\bm x, \bm y, \bm z} \bigotimes_{j=1}^k \ketbra{z_{\pi(j)}}{z_j} \cdot C_{x_j y_j} \ketbra{x_j}{y_j} = \sum_{\bm x, \bm y} \bigotimes_{j=1}^k C_{x_j y_j} \ketbra{x_{\pi(j)}}{y_j} \\
    \implies \trace{S_k \rho^{\ox k}} &= \sum_{\bm x} \prod_{j=1}^k C_{x_j x_{\pi(j)}} = \trace{ \rho^k }
.\end{align}
    Lastly, the statement $\trace{S_k^\dagger \rho^{\ox k}}$ holds from the fact that $\rho$ is Hermitian and thus
\begin{equation}
    \trace{S_k^\dagger \rho^{\ox k}} = \sum_{\bm x} \prod_{j=1}^k C_{x_j x_{\pi^{-1}(j)}} = \sum_{\bm x} \prod_{j=1}^k C_{x_{\pi^{-1}(j)} x_j}^* = \sum_{\bm x} \prod_{j=1}^k C_{x_j x_{\pi(j)}}^* = \trace{ (\rho^*)^k } = \trace{ \rho^k }.
\end{equation}
\end{proof}


\section{Proof of Theorem \ref{theorem:NS_condition}}

The theorem is established through the vectorization of observables and their effective rank. To start with, we present the definitions of these two concepts.

\begin{definition}{\rm (Vectorization ~\cite{khatri2020principles})}
The vectorization of a matrix $M = \sum_{ij} M_{ij} \ketbra{i}{j}$ is defined as
\begin{equation}
    \ket{M} \coloneqq \sum_{i, j} M_{ij} \ket{j} \ox \ket{i}.
\end{equation}
\end{definition}

\begin{definition}{\rm (Hermitian effective rank)} For a $2^n\times 2^n$ Hermitian matrix $O$ in a bipartite system $AB$, the effective rank $\cR(O)$ is defined as
\begin{equation}
    \cR(O) \coloneqq \rank(\proj{O}_A),
\end{equation}
where $\proj{O}_A = \tr_B[\proj{O}_{AB}]$.
\end{definition}

\begin{definition}{\rm (Matrix of channel ~\cite{zhao2023information})}
A quantum channel $\cN$ can be represented in the matrix form $M_\cN$ by its Kraus operators $\{E_k\}$. Given $\cN(\ketbra{i}{j}) = \sum_k E_k \ketbra{i}{j} E_k^\dagger$, we can represent it in vector form, which is $\text{vec}(\cN(\ketbra{i}{j})) = \sum_k \bar{E_k}\ket{j} \ox E_k \ket{i} = (\sum_k \bar{E_j}\ox E_k)\ket{j}\ox\ket{i}$, where $\bar{E_k}$ is the complex conjugate of $E_k$. Here we can define 
\begin{equation}
    M_\cN \coloneqq \sum_k \bar{E_k}\ox E_k
\end{equation}
as the matrix representation of channel $\cN$. The rank of such matrix $M_\cN$ is called the channel rank of $\cN$.
\end{definition}

In addition, here are a few results that will be used in the following proof of Theorem~\ref{theorem:NS_condition}.

\begin{lemma}\label{lemma:simplfy}{\rm \cite{zhao2023information}}
    Suppose $H$ is an observable and $\cM$ is a Hermitian-preserving is a map. Then $\tr[(\cM(\rho)O)]=\tr[\rho O]$ for any state $\rho$ if and only if $\cM^\dagger(O) = O$.
\end{lemma}

\begin{lemma} \label{lem:full effective rank}
    Suppose $H$ is an $nk$-qubit Hermitian matrix with $k$ subsystems, such that $\trace{H \rho^{\ox k}} = \trace{\rho^k}$ for any $n$-qubit state $\rho$. Then the effective rank of $H$ for arbitrary subsystem $A$ is full.
\end{lemma}
\begin{proof}
    The property $\trace{XY} = \braket{X}{Y}$ implies 
\begin{equation}
    \braket{\rho^{\ox k}}{H} = \trace{\rho^k}
.\end{equation}
    Let $A$ be an $n$-qubit subsystem and $B$ be the rest of the subsystems. With respect to bipartite system $AB$, applying Schmidt decomposition on $\ket{H}$ gives
\begin{equation}
    \braket{\rho^{\ox k}}{H} = \sum_{j} c_j \braket{\rho}{u_j}_A \cdot \braket{\rho^{\ox (k - 1)}}{v_j}_B \textrm{, where } \ket{H} = \sum_j c_j \ket{u_j}_A \ket{v_j}_B
.\end{equation}
    Note that $\ketbra{H}{H}_A = \sum_{j} |c_j|^2 \ketbra{u_j}{u_j}$. Then the assumption $\ket{\rho} \in \ker(\ketbra{H}{H}_A)$ would imply
\begin{equation}
    \sum_j |c_j|^2 \braket{\rho}{u_j} \braket{u_j}{\rho} = 0
    \implies \braket{\rho}{u_j} = 0 \,\,\forall j 
    \implies \braket{\rho^{\ox k}}{H} = 0 
    \implies \trace{\rho^k} = 0
,\end{equation}
    which forms a contradiction as $\trace{\rho^k} > 0$ for finite $k$. That is, the vectorization of all $n$-qubit quantum states are in the column space of $\ketbra{H}{H}_A$, and hence $\ketbra{H}{H}_A$ is of full rank.
\end{proof}

We are now ready to present the main theorem.

\vspace{0.5cm}
\noindent\textbf{Theorem ~\ref{theorem:NS_condition}} {(\rm Necessary and sufficient condition for existence of protocol)}
\textit{Given a noisy channel $\cN$, there exists a quantum protocol to extract the $k$-th moment $\tr[\rho^k]$ for any state $\rho$ if and only if the noisy channel $\cN$ is invertible.}
\vspace{0.5cm}

\begin{proof}
    From Lemma~\ref{lemma:exitance_of_H}, there exists an observable $H$ such that $\tr[H\rho^{\ox k}] = \tr[\rho^k]$. If the noisy channel $\cN$ is invertible, then there exists the inverse operation of noisy channel $\cN^{-1}$, such that we can apply the inverse operation on quantum system, and perform measurements with respect to moment observable $H$. Thus, the high-order moment can be retrieved, which is
    \begin{align}
        \tr[H \cdot (\cN^{-1})^{\ox k}\circ\cN^{\ox k}(\rho^{\ox k})] = \tr[H\rho^{\ox k}] = \tr[\rho^k].
    \end{align}
    The proof for the "if" part is completed. 

    For the "only if" part, let's assume there exists a quantum protocol to extract $k$-th order information $\tr[\rho^k]$ from noisy state $\cN(\rho)$, with non-invertible noise channel $\cN$. Then, from the definition of quantum protocol, the existence of quantum protocol to extract information $\tr[\rho^k]$ refers to the existence of an HP map $\cD$ such that the following equation holds
    \begin{equation}
        \tr[H \cdot \cD\circ\cN^{\ox k}(\rho^{\ox k})] = \tr[H \rho^{\ox k}] = \tr[\rho^ k].
    \end{equation}
    From Lemma~\ref{lemma:simplfy}, we convert the problem into that there exists an HP map $\cD$ such that the following equation holds
    \begin{equation}
        (\cN^{\ox k})^\dagger \circ \cD^\dagger (H) = H.
    \end{equation}
    Since $\cD$ is an HP map, then the adjoint map $\cD^\dagger$ is an HP map as well, thus we have $D^\dagger(H) = Q$, where $Q$ is hermitian. Now we have
    \begin{equation}
        (\cN^{\ox k})^\dagger (Q) = H
    \end{equation}

    Next we vectorize the Hermitian matrix $H$ and $Q$ into $\ket{H}$ and $\ket{Q}$, the corresponding channel $(\cN^{\ox k})^\dagger$ in vectoriztion representation is $M_{\cN^\dagger}^{\ox k}$, where $M_{\cN^\dagger}=\sum_k\bar{E_k}^\dagger\ox E_k^\dagger$, $\{E_k\}$ is the Kraus representation of channel $\cN$ operator, and $\bar{E_k}$ is the complex conjugate of $E_k$. Thus the equation becomes
    \begin{equation}
       M_{\cN^\dagger}^{\ox k}\ket{Q} = \ket{H},
    \end{equation}
    which is equivalent to 
    \begin{equation}
        M_{\cN^\dagger}^{\ox k} \ketbra{Q}{Q} M_{\cN^\dagger}^{\ox k\dagger} = \ketbra{H}{H}.
    \end{equation}
    Then take partial trace on both side,
    \begin{equation}
         \tr_B[M_{\cN^\dagger}^{\ox k} \ketbra{Q}{Q} M_{\cN^\dagger}^{\ox k \dagger}] = \tr_B[\ketbra{H}{H}]
    ,\end{equation}
    where $B$ is an arbitrary subsystem that contains $k - 1$ copies of $\rho$. We can simplify it into
    \begin{equation}
        M_{\cN^\dagger}\tr_B[ I\ox M^{\ox (k-1)}_{\cN^\dagger}\ketbra{Q}{Q}I\ox M_{\cN^\dagger}^{\ox(k-1)\dagger}] M_{\cN^\dagger}^\dagger = \tr_B[\ketbra{H}{H}].
    \end{equation}
    For further step, let's take $P = \tr_B[ I\ox M^{\ox (k-1)}_{\cN^\dagger}\ketbra{Q}{Q}I\ox M_{\cN^\dagger}^{\ox(k-1)\dagger}]$, and $T=\tr_B[\ketbra{H}{H}]$, thus
    \begin{equation}
        M_{\cN^\dagger} P M_{\cN^\dagger}^\dagger = T.
    \end{equation}
    From Lemma~\ref{lem:full effective rank}, the matrix $T$ is full rank, thus the rank of $M_{\cN^\dagger} P M_{\cN^\dagger}^\dagger$ is full as well. Due to the fact that $\rank(XY)\le \min (\rank(X), \rank(Y))$, the rank on the left hand side is
    \begin{equation}
        \rank(M_{\cN^\dagger} P M_{\cN^\dagger}^\dagger)\le\min \{ \rank(M_{\cN^\dagger}), \rank(P), \rank(M_{\cN^\dagger}^\dagger) \}
    ,\end{equation}
    implying that $M_{\cN^\dagger}$, $P$ and $M_{\cN^\dagger}^\dagger$ have full ranks. Then we have
\begin{equation}
    M_{\cN^\dagger} = \sum_k\bar{E_k}^\dagger\ox E_k^\dagger = \sum_k(\bar{E_k}\ox E_k)^\dagger = M_\cN ^\dagger
,\end{equation}
   since conjugate transpose won't affect the rank of a matrix, which means $M_\cN$ has full rank. This concludes that $\cN$ is invertible, and eventually, a contradiction is formed. Thus, no quantum protocol exists to extract high-order moment $\tr[\rho^k]$.
\end{proof}

\section{Proof for Lemma~\ref{lemma:CPTS}}

\begin{lemma}{\rm ~\cite{zhao2023information}}\label{lemma:HPTS}
    Given a quantum channel $\cN$ and an observable $O$, there exists a HPTS map $\cD$ 
    such that $\cN^\dagger\circ\cD^\dagger(O)=O$ if and only if $O\in\cN^\dagger(\cL^H_{A'})$.
\end{lemma}

\vspace{0.5cm}
\noindent\textbf{Lemma~\ref{lemma:CPTS}}
\textit{Given an invertible quantum channel $\cN$ and an observable $O$, there exists a quantum channel $\cC$ and real coefficients $t, f$ such that}
\begin{equation}
    \cN^\dagger\circ\cC^\dagger(O)=\frac{1}{f}\left(O + tI\right).
\end{equation}

\begin{proof}
    Since the noise channel $\cN$ is invertible, the image of the adjoint map of the noise channel $\cN^\dagger(\cL)$ has full dimension and hence $O\in \cN^\dagger(\cL)$. From Lemma~\ref{lemma:HPTS}, there exists a HPTS map $\cD$ such that $\cN^\dagger\circ \cD^\dagger(H)=H$. Then we take the trick of observable shift, which is
    \begin{align}
        \cN^\dagger \left(\cD^\dagger(H)\right) =H+tI-tI.
    \end{align}
    Since $\cN$ is CPTP, its adjoint map $\cN^\dagger$ is unital-preserving, implying
    \begin{equation}
        \cN^\dagger\circ \cD^\dagger(H+tI)=H+tI.
    \end{equation}
    Denote $\Tilde{\cD}^\dagger(H) = \cD^\dagger(H)+tI$. Since $\cD$ is an HPTS map, thus $\cD^\dagger$ is a HPUS, i.e., $\cD^\dagger(I) = a I$, where $a$ is a real coefficient. Correspondingly, 
    \begin{equation}
        \Tilde{\cD}^\dagger (I) = \cD^\dagger(I) + tI = (a+t)I
    \end{equation}
    is also an HPUS map. Due to the fact that the adjoint of an HPUS map is HPTS, meaning $\Tilde{\cD}$ is an HPTS map. As long as the value $t$ no smaller than the absolute minimum eigenvalue of $\cD^\dagger(H)$, i.e., $t\ge |\min\{eig(\cD^\dagger(H))\}|$, the map $\Tilde{\cD}$ reduces to CPTS. Thus, we have 
    \begin{equation}
        \cN^\dagger\circ \Tilde{\cD}^\dagger(H) =H+tI
    \end{equation}
    Note that any CPTS map and be written as a coefficient $f$ times a CPTP map $\cC$, which is $\Tilde{\cD} = f\cC$, meaning
    \begin{align}
        f\cN^\dagger\circ \cC^\dagger(H) &=H+tI\\
        \Rightarrow \cN^\dagger\circ \cC^\dagger(H) &=\frac{1}{f}(H+tI)
    \end{align}
        
    which complete the proof.
\end{proof}

\section{Dual SDP for sampling overhead}
The original SDP is 

\begin{subequations}
\begin{align}
    f_{\min}(\cN, k) &= \min  \quad f\\
    \st &\quad J_{\Tilde{\cC}} \ge 0 \\
    &\quad \tr_{C}[J_{\Tilde{\cC}_{BC}}]=f I_{B}\\
    &\quad J_{\cF_{AC}} \equiv \tr_B[(J_{\cN_{AB}^{\ox k}}^{T_B}\ox I_C) (I_A\ox J_{\Tilde{\cC}_{BC}})]\\
    &\quad \tr_C[(I_A\ox H_C^T)J_{\cF_{AC}}^T] = H_A + tI_A
\end{align}
\end{subequations}
The Lagrange function is 
\begin{align}
    L(f, J_{\Tilde{\cC}},t, M, K) &= f + \langle M, \tr_C[J_{\Tilde{\cC}}]-f I_B \rangle + \langle K, \tr_C[(I_A\ox H_C^T)J_{\cF_{AC}}^T] - H_A - tI_A \rangle\\
    &= f(1-\tr[M]) + \langle M\ox I, J_{\Tilde{\cC}} \rangle - \langle K, H\rangle - t\tr[K] + \langle K, \tr_C[(I_A\ox H_C^T)J_{\cF_{AC}}^T]\rangle,
\end{align}
where $M,K$ are the dual variables. The last term can be expanded as  
\begin{align}
    \langle K, \tr_A[(H_A^T\ox I_C)J_{\cF_{AC}}^T]\rangle &= \tr[K \tr_C[(I_A\ox H_C^T)J_{\cF_{AC}}^T]]\\
    &= \tr[K\ox I (I_A\ox H_C^T)J_{\cF_{AC}}^T]\\
    &= \tr[(K_A\ox H_C^T)J_{\cF_{AC}}^T]\\
    &= \tr[(K_A^T\ox H_C)J_{\cF_{AC}}]\\
    &= \tr[(K_A^T\ox H_C) \tr_B[(J_{\cN_{AB}^{\ox k}}^{T_B}\ox I_C) (I_A\ox J_{\Tilde{\cC}_{BC}})]]\\
    &= \tr[(K_A^T\ox I_B\ox H_C) (J_{\cN_{AB}^{\ox k}}^{T_B}\ox I_C) (I_A\ox J_{\Tilde{\cC}_{BC}})]\\
    &= \tr[\tr_A[(K_A^T\ox I_B\ox H_C)(J_{\cN_{AB}^{\ox k}}^{T_B}\ox I_C)] J_{\Tilde{\cC}_{BC}}]
\end{align}
Thus, the Lagrange function can be expressed as
\begin{align}
    \cL(f, J_{\Tilde{\cC}},t, M, K)
    &= f(1-\tr[M]) + \langle M\ox I, J_{\Tilde{\cC}} \rangle - \langle K, H\rangle - t\tr[K] + \langle \tr_A[(K_A^T\ox I_B\ox H_C)(J_{\cN_{AB}^{\ox k}}^{T_B}\ox I_C)], J_{\Tilde{\cC}_{BC}}\rangle\\
    &= f(1-\tr[M])  - \langle K, H\rangle - t\tr[K] + \langle  M\ox I + \tr_A[(K_A^T\ox I_B\ox H_C)(J_{\cN_{AB}^{\ox k}}^{T_B}\ox I_C)], J_{\Tilde{\cC}}\rangle
\end{align}
The corresponding Lagrange dual function is
\begin{align}
    f(M,K) &= \inf_{J_{\Tilde{\cD}}, t, f} \cL(f,J_{\Tilde{\cC}},t, M,K)\\
    &=  - \langle K, H\rangle + \inf_{J_{\Tilde{\cD}}, t, f}( f(1-\tr[M]) - t\tr[K] + \langle  M\ox I + \tr_A[(K_A^T\ox I_B\ox H_C)(J_{\cN_{AB}^{\ox k}}^{T_B}\ox I_C)], J_{\Tilde{\cC}}\rangle)
\end{align}
In order to bound the infimum term, every sub-term inside the infimum should greater greater or equal to 0, thus, we have 
\begin{align}
    &1-\tr[M]\ge 0, \quad t\tr[K] = 0\\
    &\langle  M\ox I + \tr_A[(K_A^T\ox I_B\ox H_C)(J_{\cN_{AB}^{\ox k}}^{T_B}\ox I_C)], J_{\Tilde{\cC}}\rangle \ge 0
\end{align}
Since $J_{\Tilde{\cC}} \ge 0$, we have $M\ox I + \tr_A[(K_A^T\ox I_B\ox H_C)(J_{\cN_{AB}^{\ox k}}^{T_B}\ox I_C)]\ge 0$. Thus, we arrive at the following dual SDP: 
\begin{subequations}\label{eq:dual}
\begin{align}
    f_{\min}(\cN, k) &= \max  \quad -\tr[KH]\\
    \st 
    &\quad \tr[M] \le 1\\
    &\quad \tr[K] = 0\\
    &\quad M\ox I + \tr_A[(K_A^T\ox I_B\ox H_C)(J_{\cN_{AB}^{\ox k}}^{T_B}\ox I_C)]\ge 0
\end{align}
\end{subequations}

\section{Optimal sampling overhead for depolarizing channel (Proposition~\ref{prop:retrieve_DE})}

The depolarizing channel is one of the most common noises in quantum computers which maps part of the quantum state into a maximally mixed state. The $1$-qubit depolarizing channel can be written in the following form:
\begin{equation}
    \cN^\epsilon_{\rm DE}(\rho) = (1-\epsilon)\rho + \epsilon\frac{I}{2},
\end{equation}

where $\epsilon\in [0,1]$ is the noisy level and $I$ is identity matrix. The retrieving protocol and retrieving overhead for such noise channel was well studied in ~\cite{zhao2023information}. Here we will focus on retrieving the second order information $\tr[\rho^2]$ from two-copy of depolarizing noised state $\cN(\rho\ox\rho) = \cN^\epsilon_{\rm DE}\ox \cN^\epsilon_{\rm DE}(\rho\ox\rho)$, after applying such channel, the state becomes
\begin{align}
    \cN(\rho\ox\rho) &= \cN^\epsilon_{\rm DE}(\rho)\ox\cN^\epsilon_{\rm DE}(\rho) \\
    &= [(1-\epsilon) \rho + \epsilon \frac{I}{2}] \ox [(1-\epsilon) \rho + \epsilon \frac{I}{2}]\\ 
    &= (1-\epsilon)^2 \rho\ox\rho + \frac{\epsilon(1-\epsilon)}{2} I\ox\rho + \frac{(1-\epsilon) \epsilon}{2} \rho\ox I + \frac{\epsilon^2}{4} I\ox I.\label{eq:depo_state}
\end{align}

\vspace{0.5cm}
\noindent\textbf{Proposition~\ref{prop:retrieve_DE}} 
\textit{Given noisy states $\cN^\epsilon_{\rm DE}(\rho)^{\ox 2}$, and error tolerance $\delta$, the second order moment $\tr[\rho^2]$ can be estimated by $f\trace{H\cC\left(\cN^\epsilon_{\rm DE}(\rho)^{\ox 2}\right)} - t$, with optimal sample complexity $\mathcal{O}(1/(\delta^2 (1-\epsilon)^4)$, where $f=\frac{1}{(1-\epsilon)^2}$, $t = \frac{1-(1-\epsilon)^2}{2(1-\epsilon)^2}$. The term $\trace{H\cC\left(\cN^\epsilon_{\rm DE}(\rho)^{\ox 2}\right)}$ can be estimated by implementing a quantum retriever $\cC$ on noisy states and making measurements over moment observable $H$. Moreover, there exists an ensemble of unitary operations $\{p_j, U_j\}_j$ such that the action of the retriever $\cC$ can be interpreted as,}
\begin{equation}
    \cC(\cdot) = \sum_{j=1} p_j 
    U_j \cdot U_j^\dagger
\end{equation}
    
\begin{proof}
    To begin with, we give the explicit form of the retriever $\cC$. 
    \begin{equation}
        \cC(\cdot) = \sum_{j=1}^{12} \frac{1}{12} 
        U_j \cdot U_j^\dagger,
    \end{equation}
    where all the probabilities $p_j$ are the same, i.e., $p_j=\frac{1}{12}$, and the corresponding unitaries $U_j$ are 
    \begin{equation}
    \begin{aligned}
        U_1 &= I\ox I; \quad U_2 = X\ox X; \quad U_3 = Y\ox Y; \quad U_4 = Z\ox Z;\\
        U_5 &= \frac{1}{2}
        \begin{pmatrix}
            -i & 1 & 1 & i\\
            i & 1 & -1 & i\\
            i & -1 & 1 & i\\
            -i & -1 & -1 & i
        \end{pmatrix}; \
        U_6 = \frac{1}{2}
        \begin{pmatrix}
            i & -1 & -1 & -i\\
            i & 1 & -1 & i\\
            i & -1 & 1 & i\\
            i & 1 & 1 & -i
        \end{pmatrix}; 
        \ U_7 = \frac{1}{2}
        \begin{pmatrix}
            i & i & i & i\\
            -1 & 1 & -1 & 1\\
            -1 & -1 & 1 & 1\\
            -i & i & i & -i
        \end{pmatrix};
        U_8 = \frac{1}{2}
        \begin{pmatrix}
            -i & i & i & -i\\
            1 & 1 & -1 & -1\\
            1 & -1 & 1 & -1\\
            i & i & i & i
        \end{pmatrix}; \\
        U_9 &= \frac{1}{2}
        \begin{pmatrix}
            i & 1 & 1 & -i\\
            -i & 1 & -1 & -i\\
            -i & -1 & 1 & -i\\
            i & -1 & -1 & -i
        \end{pmatrix}; \ U_{10} = \frac{1}{2}
        \begin{pmatrix}
            -i & -1 & -1 & i\\
            -i & 1 & -1 & -i\\
            -i & -1 & 1 & -i\\
            -i & 1 & 1 & i
        \end{pmatrix};
        U_{11} = \frac{1}{2}
        \begin{pmatrix}
            i & -i & -i & i\\
            1 & 1 & -1 & -1\\
            1 & -1 & 1 & -1\\
            -i & -i & -i & -i
        \end{pmatrix}; \
        U_{12} = \frac{1}{2}
        \begin{pmatrix}
            -i & -i & -i & -i\\
            -1 & 1 & -1 & 1\\
            -1 & -1 & 1 & 1\\
            i & -i & -i & i
        \end{pmatrix}.
    \end{aligned}
    \end{equation}
    
    Next, we are going to prove the retriever $\cC$ has the optimal sampling overhead. The optimal complexity $\mathcal{O}(1/(\delta^2 (1-\epsilon)^4)$ implies the optimal sampling overhead is $\frac{1}{(1-\epsilon)^2}$. To prove estimating the second order moment by implementing quantum channel $C$ with optimal sampling overhead is $\frac{1}{(1-\epsilon)^2}$, we divide the process into two steps. For the first step, we prove that quantum channel $\cC$ is a feasible solution to the prime problem with sampling overhead $\frac{1}{(1-\epsilon)^2}$, meaning the optimal overhead $\gamma^*$, $\gamma^* \le\frac{1}{(1-\epsilon)^2}$. Then we are going to find one feasible solution to the dual problem with overhead $\frac{1}{(1-\epsilon)^2}$, meaning $\gamma^* \ge\frac{1}{(1-\epsilon)^2}$. Thus we have the optimal overhead $\gamma^*=\frac{1}{(1-\epsilon)^2}$, which complete the proof.

    Now, we are going to proof the first part. At first, we need to write the Choi matrix form of the retriever $\cC$, which is 
    \begin{equation}
        J_\cC = \frac{1}{4} IIII + \frac{1}{12}(XX+YY+ZZ)\ox(XX+YY+ZZ).
    \end{equation}
    For an arbitrary state $\rho\ox\rho$, after applying the noise $\cN$ onto it, we denote the state as $\rho' = \cN(\rho\ox\rho)$, thus 
    \begin{align}
        \cC\circ \cN(\rho\ox\rho) &= \cC(\rho')  = \tr_A[(\rho'^T \ox II) J_\cC]\\
        &=  \tr_A\Big[(\rho'^T \ox II) (\frac{1}{4} IIII  + \frac{1}{12}(XX+YY+ZZ)\ox(XX+YY+ZZ))\\
        &=  \frac{1}{4}\tr[\rho'^T] II
       + \frac{1}{12}\tr[\rho'^T(XX+YY+ZZ)](XX+YY+ZZ))\\
        &=  \frac{1}{4}II +  \frac{1}{12}\tr[\rho'^T(XX+YY+ZZ)](XX+YY+ZZ)) \Big]\label{eq:D_circ_N_depo}
    \end{align}
    Note that the above equations utilized the face that transpose operation is trace preserving, i.e., $\tr[\rho^T]=\tr[\rho]=1$. Since the matrix $XX+YY+ZZ$ is symmetry, thus we have
    \begin{align}
        \cC\circ \cN(\rho\ox\rho) &= \frac{1}{4}II +  \frac{1}{12} \tr[\rho'^T(XX+YY+ZZ)^T](XX+YY+ZZ))\\
        &= \frac{1}{4}II +  \frac{1}{12} \tr[\rho'(XX+YY+ZZ)](XX+YY+ZZ)) \label{eq:DE_after_transpose}
    \end{align}
    The trace term can be calculated by substituting Eq.~\eqref{eq:depo_state}, we have 
    \begin{align}
        \tr[\rho'(XX+YY+ZZ)] &= \tr[((1-\epsilon)^2 \rho\ox\rho + \frac{\epsilon(1-\epsilon)}{2} I\ox\rho + \frac{(1-\epsilon) \epsilon}{2} \rho\ox I + \epsilon^2 (I\ox I)) (XX+YY+ZZ)]\\
        &=(1-\epsilon)^2\tr[\rho\ox\rho (XX+YY+ZZ)] + \frac{\epsilon(1-\epsilon)}{2}\tr[X\ox \rho X + Y\ox \rho Y + Z\ox \rho Z]\nonumber\\
        &\quad +\frac{(1-\epsilon)\epsilon}{2}\tr[\rho X\ox X + \rho Y\ox Y + \rho Z\ox Z] + \frac{\epsilon^2}{4}\tr[XX+YY+ZZ]\\
        &= (1-\epsilon)^2\tr[\rho\ox\rho (XX+YY+ZZ)] .\label{eq:depo_mid}
    \end{align}
    In the second equation, since $X,Y,Z$ are all traceless Hermitian matrices, all terms are zeros except the first term. Replace the equation back to Eq.~\eqref{eq:DE_after_transpose}, then
    \begin{align}
        \cC\circ \cN(\rho\ox\rho) &=  \frac{1}{4}II + \frac{(1-\epsilon)^2}{12}\tr[\rho\ox\rho (XX+YY+ZZ)] (XX+YY+ZZ).
    \end{align}
    The information $\tr[\rho^2]$ is estimated from $ \tr[H(\rho\ox\rho)]$, where $H = \frac{1}{2}(II+XX+YY+ZZ)$ is cyclic permutation operator (in the 2-qubit case, $H$ is just SWAP operator). It is easy to check that 
    \begin{align}
        \tr[H \cC\circ\cN(\rho\ox\rho)] &= \frac{1}{(1-\epsilon)^2} \Big[\frac{1}{4}\tr[H*II] + \frac{(1-\epsilon)^2}{12}\tr[\rho\ox\rho (XX+YY+ZZ)]\tr[H*(XX+YY+ZZ)] \Big].
    \end{align}
    We can quickly get $\tr[H] = 2$ and $\tr[H*(XX+YY+ZZ)]=6$. Then
    \begin{align}
        \tr[H \cC\circ\cN(\rho\ox\rho)] &= \frac{1}{2} +\frac{(1-\epsilon)^2}{2}\tr[\rho\ox\rho (XX+YY+ZZ)]\\
        &= (1-\epsilon)^2 \Big[\frac{1}{2(1-\epsilon)^2} + \frac{1}{2}\tr[\rho\ox\rho (XX+YY+ZZ)]\Big]\\
        &= (1-\epsilon)^2 \Big[\frac{1}{2(1-\epsilon)^2}-\frac{1}{2} + \frac{1}{2} +  \frac{1}{2}\tr[\rho\ox\rho (XX+YY+ZZ)]\Big]\\
        &= (1-\epsilon)^2 \Big[\frac{1}{2(1-\epsilon)^2}-\frac{(1-\epsilon)^2}{2(1-\epsilon)^2} +  \frac{1}{2}\tr[\rho\ox\rho (XX+YY+ZZ+II)]\Big]\\
        &= (1-\epsilon)^2 \Big[\frac{2\epsilon^2-\epsilon^2}{2(1-\epsilon)^2} + \tr[\rho\ox\rho H]\Big]\\
        &= (1-\epsilon)^2 \Big[\frac{2\epsilon^2-\epsilon^2}{2(1-\epsilon)^2} + \tr[\rho^2]\Big].
    \end{align}
    The desired high-order moment $\tr[\rho^2]$ value equals to
    \begin{equation}
        \tr[\rho^2] = f\tr[H \cC\circ\cN(\rho\ox\rho)] - t
    \end{equation}
    where $f = \frac{1}{(1-\epsilon)^2}$ is the sampling overhead and $t = \frac{1-(1-\epsilon)^2}{2(1-\epsilon)^2}$ is the shifted distance. In order to estimate the value with the error $\delta$, the sampling overhead should be $1/(1-\epsilon)^2$. Thus, we have $\gamma^*\le 1/(1-\epsilon)^2$.

    Next, we are going to use dual SDP to show that $\gamma^*\ge 1/(1-\epsilon)^2$. We set the dual variables as $M=\frac{1}{4}II-\frac{1}{12}(XX+YY+ZZ)$, and $K=q(XX+YY+ZZ)$ where $q=-\frac{1}{6(1-\epsilon)^2}$. We will show the variables $\{M,K\}$ is a feasible solution to the dual problem.

    If we substitute the variables into the dual problem Eq.~\eqref{eq:dual}, we can easily check that $\tr[M]\le 1, \tr[K]=0$. For the last condition, we have 
    \begin{align}
        \tr_A[(K_A^T\ox I_B\ox H_C)(J_{\cN_{AB}^{\ox k}}^{T_B}\ox I_C)] = -\frac{1}{12}(XX+YY+ZZ)(II+XX+YY+ZZ)
    \end{align}
    Thus, we have 
    \begin{align}
        &M\ox II + \tr_A[(K_A^T\ox I_B\ox H_C)(J_{\cN_{AB}^{\ox k}}^{T_B}\ox I_C)]\\
        &= \frac{1}{4}IIII-\frac{1}{6}(XX+YY+ZZ)II - \frac{1}{12}(XX+YY+ZZ)(XX+YY+ZZ)\\
        &\ge 0,
    \end{align}
    which means $\{M,K\}$ is a feasible solution to the dual SDP. Therefore, we have $\gamma*\ge -\tr[KH]=\frac{1}{(1-\epsilon)^2}$. Combined with prime part, we have $\gamma^* = \frac{1}{(1-\epsilon)^2}$, meaning that given error tolerance $\delta$, the optimal sample complexity is $\cO(\frac{1}{\delta^2(1-\epsilon)^4})$. The proof is complete.
    \end{proof}


\section{Optimal sampling overhead for amplitude damping channel (Proposition~\ref{prop:retrieve_AD})}
The amplitude damping (AD) channel is a physical channel that describes the energy leakage, dropping from a high energy state to a low energy state. The qubit amplitude damping channel $\cN_{\rm AD}^\varepsilon$ has two Kraus operators: $A_0^\varepsilon \coloneqq \proj{0} + \sqrt{1-\varepsilon}\proj{1}$ and $A_1^\varepsilon \coloneqq \sqrt{\varepsilon}\ketbra{0}{1}$. A single qubit state $\begin{pmatrix}
    \rho_{00} & \rho_{01}\\
    \rho_{10} & \rho_{11}
\end{pmatrix}$
after going through the AD channel is 
\begin{equation}
    \cN_{\rm AD}^\varepsilon(\rho)=\begin{pmatrix}
    \rho_{00}+\varepsilon\rho_{11} & \sqrt{1-\varepsilon}\rho_{01}\\
    \sqrt{1-\varepsilon}\rho_{10} & (1-\varepsilon)\rho_{11}
    \end{pmatrix}
\end{equation}
where $\varepsilon\in[0,1]$ is the damping factor. This part we will focus on retrieving the information $\tr[\rho^2]$ from two-copy AD channel noised state $\cN_{\rm AD}^\varepsilon\ox\cN_{\rm AD}^\varepsilon(\rho\ox\rho)$. The noised state is
\begin{align}
    \cN(\rho\ox\rho) &= \cN_{\rm AD}^\varepsilon\ox\cN_{\rm AD}^\varepsilon(\rho\ox\rho)\\
    &=\begin{pmatrix}
        (\rho_{00}+\varepsilon\rho_{11})^2& (\rho_{00}+\varepsilon\rho_{11})(\sqrt{1-\varepsilon}\rho_{01}) & (\sqrt{1-\varepsilon}\rho_{01})(\rho_{00}+\varepsilon\rho_{11})& (1-\varepsilon)\rho_{01}^2 \\
        (\rho_{00}+\varepsilon\rho_{11})(\sqrt{1-\varepsilon}\rho_{10}) &(\rho_{00}+\varepsilon\rho_{11})(1-\varepsilon)\rho_{11} & (1-\varepsilon) \rho_{01}\rho_{10} & (1-\varepsilon)\sqrt{1-\varepsilon}\rho_{01}\rho_{11}\\
        \sqrt{1-\varepsilon}\rho_{10}(\rho_{00}+\varepsilon\rho_{11})&(1-\varepsilon)\rho_{10}\rho_{01} & (1-\varepsilon)\rho_{11} (\rho_{00}+\varepsilon\rho_{11})& (1-\varepsilon)\sqrt{1-\varepsilon}\rho_{11}\rho_{01}\\
        (1-\varepsilon)\rho_{10}^2 & (1-\varepsilon)\sqrt{1-\varepsilon}\rho_{10} \rho_{11} & (1-\varepsilon)\sqrt{1-\varepsilon}\rho_{11}\rho_{10}&(1-\varepsilon)^2\rho_{11}^2\\
    \end{pmatrix}\label{eq:AD_matrix}
\end{align} 

\vspace{0.5cm}
\noindent\textbf{Proposition \ref{prop:retrieve_AD}} 
\textit{Given noisy states $\cN^\varepsilon_{\rm AD} (\rho)^{\ox 2}$, and error tolerance $\delta$, the second order moment $\tr[\rho^2]$ can be estimated by $f\trace{H\cC\left(\cN^{\varepsilon}_{\rm AD}(\rho)^{\ox 2}\right)} - t$, with optimal sample complexity $\mathcal{O}(1/(\delta^2 (1-\varepsilon)^4)$, where $f=\frac{1}{(1-\varepsilon)^2}$, $t = -\frac{\varepsilon^2}{(1-\varepsilon)^2}$. The term $\trace{H\cC\left(\cN^{\varepsilon}_{\rm AD}(\rho)^{\ox 2}\right)}$ can be estimated by implementing a quantum retriever $\cC$ on noisy states and making measurements. Moreover, the Choi matrix of such the retriever $\cC$ is}
    \begin{align}
        J_\cC &= \proj{00} \ox \frac{1}{6} ((1+2\varepsilon)II + (1-4\varepsilon)H) \nonumber\\
        &\quad + \proj{\Psi^+} \ox \frac{1}{6}((1+2\varepsilon)II + (1-4\varepsilon)H) \nonumber\\
        &\quad + \proj{\Psi^-} \ox \frac{1}{2} (II - H) \nonumber\\
        &\quad + \proj{11} \ox \frac{1}{6}(II + H).
    \end{align}
    \textit{where $\proj{\Psi^\pm} = \frac{1}{2}(\ket{01} \pm \ket{10})(\bra{01} \pm \bra{10})$ are Bell states.}
  
\begin{proof}
    The optimal complexity $\mathcal{O}(1/(\delta^2 (1-\varepsilon)^4)$ implies the optimal sampling overhead is $\frac{1}{(1-\varepsilon)^2}$. To prove to estimate the second order moment by implementing quantum channel $C$ with optimal sampling overhead is $\frac{1}{(1-\varepsilon)^2}$, we divide the process into two steps. For the first step, we prove that quantum channel $\cC$ is a feasible solution to the prime problem with sampling overhead $\frac{1}{(1-\varepsilon)^2}$, meaning the optimal overhead $\gamma^*$, $\gamma^* \le\frac{1}{(1-\varepsilon)^2}$. Then we are going to find one feasible solution to the dual problem with overhead $\frac{1}{(1-\varepsilon)^2}$, meaning $\gamma^* \ge\frac{1}{(1-\varepsilon)^2}$. Thus we have the optimal overhead $\gamma^*=\frac{1}{(1-\varepsilon)^2}$, which complete the proof.

    Now, we are going to prove the first part. For an arbitrary state $\rho\ox\rho$, after applying the noise $\cN$ onto it, we denote the state as $\rho' = \cN(\rho\ox\rho)$, thus 
    \begin{align}\label{eq:apply_ad_decoder}
        \cC\circ \cN(\rho\ox\rho) &= \cC(\rho')  = \tr_A[(\rho'^T \ox II) J_\cC]\nonumber\\
        & =\tr_A\Big[(\rho'^T \ox II)  \Big( \proj{00} \ox \frac{1}{6} ((1+2\varepsilon)II + (1-4\varepsilon)H) \nonumber\\
        &\quad + \frac{1}{2}(\ket{01} + \ket{10})(\bra{01} + \bra{10}) \ox \frac{1}{6}((1+2\varepsilon)II + (1-4\varepsilon)H) \nonumber\\
        &\quad + \frac{1}{2}(\ket{01} - \ket{10})(\bra{01} - \bra{10}) \ox \frac{1}{2} (II - H) \nonumber\\
        &\quad + \proj{11} \ox \frac{1}{6}(II + H) \Big)\Big]\\
        &= \tr[\rho'^T \proj{00}]\bigg(\frac{1}{6} ((1+2\varepsilon)II + (1-4\varepsilon)H)\bigg)\nonumber\\
        &\quad + \tr[\rho'^T ((\ket{01} + \ket{10})(\bra{01} + \bra{10})] \bigg(\frac{1}{12}((1+2\varepsilon)II + (1-4\varepsilon)H)\bigg)\nonumber\\
        &\quad+ \tr[\rho'^T ((\ket{01} - \ket{10})(\bra{01} - \bra{10})] \bigg(\frac{1}{4}(II - H)\bigg)\nonumber\\
        &\quad+ \tr[\rho'^T \proj{11}] \bigg(\frac{1}{6}(II+H)\bigg)
    \end{align}
    we can take transpose of the quantum noisy state $\rho'$, as shown in Eq.~\eqref{eq:AD_matrix}, then substitute it into our equation, we get 
    \begin{align}
        \cC\circ \cN(\rho\ox\rho)
        &= (\rho_{00}+\varepsilon\rho_{11})^2\bigg(\frac{1}{6} ((1+2\varepsilon)II + (1-4\varepsilon)H)\bigg)\nonumber\\
        &\quad + [(1-\varepsilon)\rho_{11}(\rho_{00}+\varepsilon\rho_{11}) + (1-\varepsilon)\rho_{01}\rho_{10}] \bigg(\frac{1}{6}((1+2\varepsilon)II + (1-4\varepsilon)H)\bigg)\nonumber\\
        &\quad+ [(1-\varepsilon)\rho_{11}(\rho_{00}+\varepsilon\rho_{11}) - (1-\varepsilon)\rho_{01}\rho_{10}] \bigg(\frac{1}{2}(II - H)\bigg)\nonumber\\
        &\quad+ (1-\varepsilon)^2\rho_{11}^2 \bigg(\frac{1}{6}(II+H)\bigg)
    \end{align}
    The value $\tr[\rho^2]$ is usually estimated by $\tr[H(\rho\ox\rho)]$, where $H=\frac{1}{2}(II+XX+YY+ZZ)$ is the two-qubit cyclic permutation operator. Then the estimated value from our protocol can be arrived at
    \begin{align}
        \cC\circ \cN(\rho\ox\rho)
        &= (\rho_{00}+\varepsilon\rho_{11})^2\bigg(\frac{1}{6} \tr[H((1+2\varepsilon)II + (1-4\varepsilon)H)]\bigg)\nonumber\\
        &\quad + [(1-\varepsilon)\rho_{11}(\rho_{00}+\varepsilon\rho_{11}) + (1-\varepsilon)\rho_{01}\rho_{10}] \bigg(\frac{1}{6}\tr[H((1+2\varepsilon)II + (1-4\varepsilon)H)]\bigg)\nonumber\\
        &\quad+ [(1-\varepsilon)\rho_{11}(\rho_{00}+\varepsilon\rho_{11}) - (1-\varepsilon)\rho_{01}\rho_{10}] \bigg(\frac{1}{2}\tr[H(II - H)]\bigg)\nonumber\\
        &\quad+ (1-\varepsilon)^2\rho_{11}^2 \bigg(\frac{1}{6} \tr[H(II+H)]\bigg)
    \end{align}

    Note that the $\tr[H]=2$ and $\tr[H\cdot H]=4$, then we have
    \begin{align}
        \tr[H \cC\circ\cN(\rho\ox\rho)]
        &= (\rho_{00}+\varepsilon\rho_{11})^2(1-2\varepsilon)) + [(1-\varepsilon)\rho_{11}(\rho_{00}+\varepsilon\rho_{11}) + (1-\varepsilon)\rho_{01}\rho_{10}] (1-2\varepsilon)\nonumber\\
        &\quad - [(1-\varepsilon)\rho_{11}(\rho_{00}+\varepsilon\rho_{11}) - (1-\varepsilon)\rho_{01}\rho_{10}] + (1-\varepsilon)^2\rho_{11}^2\\
        &= (1-2\varepsilon)\rho_{00}^2 + 2(1-\varepsilon)^2\rho_{01}\rho_{10} - 2\varepsilon^2 \rho_{00}\rho_{11} + (1-2\varepsilon)\rho_{11}^2\\
        &= (1-\varepsilon)^2 \rho_{00}^2 + 2(1-\varepsilon)^2\rho_{01}\rho_{10} + (1-\varepsilon)^2\rho_{11}^2 - \varepsilon^2\rho_{00}^2 + 2\varepsilon^2 \rho_{00}\rho_{11} + \varepsilon^2\rho_{11}^2\\
        &= (1-\varepsilon)^2\Big[\rho_{00}^2+ \rho_{11}^2 + 2\rho_{10}\rho_{01} -\frac{\varepsilon^2}{(1-\varepsilon)^2} [\rho_{00}^2+ \rho_{11}^2 + \rho_{00}\rho_{11}+ \rho_{11}\rho_{00}]\Big]\\
        &= (1-\varepsilon)^2\Big[\tr[H\rho\ox\rho] - \frac{\varepsilon^2}{(1-\varepsilon)^2}\Big]\\
        &= (1-\varepsilon)^2 \Big[\tr[\rho^2] - \frac{\varepsilon^2}{(1-\varepsilon)^2}\Big].\\
    \end{align}
    The desired high-order moment $\tr[\rho^2]$ value equals to 
    \begin{equation}
        \tr[\rho^2] = f\tr[H \cC\circ\cN(\rho\ox\rho)] - t
    \end{equation}
    where $f=\frac{1}{(1-\varepsilon)^2}$ is the sampling overhead and $t = - \frac{\varepsilon^2}{(1-\varepsilon)^2}$ is the shifted distance. In order to estimate the value with the error $\delta$, the sampling overhead should be $1/(1-\varepsilon)^2$. Thus, we have $\gamma^*\le 1/(1-\varepsilon)^2$

    Next, we are going to use dual SDP to show that $\gamma^*\ge 1/(1-\varepsilon)^2$. We set the dual variables as 
    \begin{align}
        M&=\frac{1}{4}(\ket{01}-\ket{10})(\bra{01}-\bra{10}) + \frac{1}{2}\proj{11}\\
        K &=\frac{1}{2(1-\varepsilon)^2} \Big[-\varepsilon\proj{00}-\proj{11}+\frac{1+\varepsilon}{2}(\proj{01}+\proj{10}) + \frac{\varepsilon-1}{2}(\ketbra{01}{10} + \ketbra{10}{01})\Big].
    \end{align}
    We will show the variables $\{M,K\}$ is a feasible solution to the dual problem. If we substitute the variables into the dual problem Eq.~\eqref{eq:dual}, we can easily check that $\tr[M]\le 1$ and $\tr[K]=1$. For the last condition, after simplifying, we have 
    \begin{equation}
        \tr_A[(K_A^T\ox I_B\ox H_C)(J_{\cN_{AB}^{\ox k}}^{T_B}\ox I_C)] = (M-\proj{11})\ox H
    \end{equation}
    Thus
    \begin{align}
        &M\ox I + \tr_A[(K_A^T\ox I_B\ox H_C)(J_{\cN_{AB}^{\ox k}}^{T_B}\ox I_C)]\\
        =&  M\ox I  + (M-\proj{11})\ox H\\
        \ge& 0 
    \end{align}
    which means $\{M,K\}$ is a feasible solution to the dual SDP, therefore, we have $\gamma^*\ge -\tr[KH]=\frac{1}{(1-\varepsilon)^2}$. Combined with prime part, we have $\gamma^*=\frac{1}{(1-\varepsilon)^2}$, meaning that given error tolerance $\delta$, the optimal sample complexity is $\cO(\frac{1}{\delta^2(1-\varepsilon)^2})$. The proof is complete.
\end{proof}

Notice in Eq.~\eqref{eq:apply_ad_decoder}, for any input state $\rho'$ to the AD-decoder $\cC$, the action can be implemented as a construction of a state,
    \begin{equation}\label{eq:decompose_ad_decoder}
    \begin{aligned}
        \cC(\rho') &= \frac{p_1}{6}\left((1+2\varepsilon)II + (1-4\varepsilon)H\right) + \frac{p_2}{6}\left((1+2\varepsilon)II + (1-4\varepsilon)H\right) + \frac{p_3}{2}(II-H) + \frac{p_4}{6}(II + H),
    \end{aligned}
    \end{equation}
    where $p_{1,2,3,4}$ can be considered as the measurement probabilities that state $\rho'$ could project onto the orthonormal basis states $\ket{00}, \ket{11}, \ket{\Psi^+}, \ket{\Psi^-}$, respectively, where $\ket{\Psi^\pm} = 1/\sqrt{2} (\ket{01}\pm\ket{10})$, and the four states form an orthonormal basis $\cB$. Since $\cC$ is CPTP for $\varepsilon\in[0,1]$. The output state can then be derived as a convex combination of states,
    \begin{equation}
        \sigma_1 = \sigma_2 = \frac{1}{6}((1+2\varepsilon)II + (1-4\varepsilon)H); \ \sigma_3 = \frac{1}{2}(II-H); \ \sigma_4 = \frac{1}{6}(II+H),
    \end{equation}
    which is just,
    \begin{equation}
        \cC(\rho') = \sum_{j=1}^4 p_j \sigma_j.
    \end{equation}
    Now the question becomes how to construct this ensemble of states $\sigma_j$'s based on $\cC(\rho')$. Notice that for states $\sigma_{1,2}$, we have,
    \begin{equation}
        \sigma_1 = \sigma_2 = \frac{1-\varepsilon}{3}\ketbra{00}{00} + \frac{1-\varepsilon}{3}\ketbra{11}{11} + \frac{1-\varepsilon}{3}\ketbra{\Psi^+}{\Psi^+} + \varepsilon \ketbra{\Psi^-}{\Psi^-},
    \end{equation}
    and,
    \begin{equation}
        \sigma_3 = \ketbra{\Psi^-}{\Psi^-}; \ \sigma_4 = \frac{1}{3}\ketbra{00}{00} + \frac{1}{3}\ketbra{11}{11} + \frac{1}{3}\ketbra{\Psi^+}{\Psi^+}.
    \end{equation}
    Notice that for these four states, the corresponding expectation value with respect to the SWAP observable can be directly computed as
    \begin{equation}
        \tr[\sigma_{1,2} H] = 1-2\varepsilon; \ \tr[\sigma_3 H] = -1; \ \tr[\sigma_4 H] = 1.
    \end{equation}
    \begin{figure}[h]
        \centering
        \includegraphics[width=0.45\linewidth]{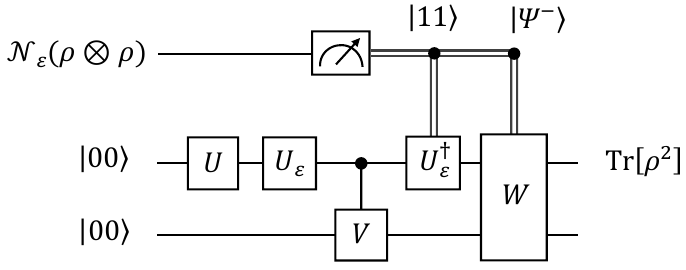}
        \caption{The circuit protocol for retrieving the 2nd moment $\tr[\rho^2]$ from AD product channel with identical noise factor $\varepsilon$. The input state $\rho' = \cN_{\varepsilon}(\rho\ox\rho)$ is the noise resource which is measured in the basis $\cB = \{\ket{00}, \ket{11}, \ket{\Psi^+}, \ket{\Psi^-} \}$. Here we need to include the ancillary system in the bottom 2-qubit system and use measurement-controlled operations shown in the diagram. The middle circuit of 2-qubit then can be used to retrieve the purity $\tr[\rho^2]$.}
        \label{fig:ad_decoder}
    \end{figure}
    
    Notice that the four measurements, in fact, form orthogonal projective measurements. The retrieving of $\tr[\rho^2]$ can then be performed via the probabilistic weighted sum of these four values based on the measurement outcomes or equivalently realized by a circuit framework as shown in Fig.~\ref{fig:ad_decoder}. In this example, the protocol can be implemented by setting up an, in total, $6$-qubit system where the top two are the input subsystems, the middle two as the retrieved output subsystems and the bottom two as the ancillary systems. We first require a $2$-qubit unitary $U$ mapping the initial states $\ket{00}$ in the middle to an intermediate state,
    \begin{equation}
        U\ket{00} = \sqrt{\frac{1}{3}}\ket{00} + \sqrt{\frac{1}{3}}\ket{11} + \sqrt{\frac{1}{3}}\ket{\Psi^+}.
    \end{equation}
    Later, a parameterized unitary $U_\varepsilon$ based on the noise factor $\varepsilon$ is settled to create the correct state coefficients, i.e.,
    \begin{equation}\label{eq:coefficient_ad}
        U_\varepsilon U\ket{00} = \sqrt{\frac{1-\varepsilon}{3}}\ket{00}+\sqrt{\frac{1-\varepsilon}{3}} \ket{11} + \sqrt{\frac{1-\varepsilon}{3}} \ket{\Psi^+} + \sqrt{\varepsilon} \ket{\Psi^-}.
    \end{equation}
    By writing out $U_\varepsilon$ with respect to the basis $\cB$, we could derive the following exact form as,
    \begin{equation}
    \begin{aligned}
        U_{\varepsilon} = \begin{pmatrix}
            1 & 0 & 0 & 0\\
            0 & \frac{3+\sqrt{3}}{6} & \frac{3-\sqrt{3}}{6} & \frac{\sqrt{3}}{3}\\
            0 & \frac{3-\sqrt{3}}{6} & \frac{3+\sqrt{3}}{6} & -\frac{\sqrt{3}}{3}\\
            0 & -\frac{\sqrt{3}}{3} & \frac{\sqrt{3}}{3} & \frac{\sqrt{3}}{3}
        \end{pmatrix}
        \begin{pmatrix}
            \sqrt{1-\varepsilon} & 0 & 0 & -\sqrt{\varepsilon}\\
            0 & \sqrt{1-\varepsilon} & -\sqrt{\varepsilon} & 0\\
            0 & \sqrt{\varepsilon} & \sqrt{1-\varepsilon} & 0\\
            \sqrt{\varepsilon} & 0 & 0 & \sqrt{1-\varepsilon}
        \end{pmatrix},
    \end{aligned}
    \end{equation}
    where in this $\cB$, the basis states are represented as,
    \begin{equation}
        \ket{00} = 
        \begin{pmatrix}
            1\\
            0\\
            0\\
            0
        \end{pmatrix}; \
        \ket{11} = 
        \begin{pmatrix}
            0\\
            1\\
            0\\
            0
        \end{pmatrix}; \
        \ket{\Psi^+} = 
        \begin{pmatrix}
            0\\
            0\\
            1\\
            0
        \end{pmatrix}; \
        \ket{\Psi^-} = 
        \begin{pmatrix}
            0\\
            0\\
            0\\
            1
        \end{pmatrix}.
    \end{equation}
    After matching the coefficients, a specialized control-$V$ ($CV$) gate will be applied to generate the purification of ~\eqref{eq:coefficient_ad} as,
    \begin{equation}
        \ket{\psi'} = CV (U_\varepsilon U \ox I)\ket{00}\ket{00} = \sqrt{\frac{1-\varepsilon}{3}}\ket{00}\ket{00}+\sqrt{\frac{1-\varepsilon}{3}} \ket{11}\ket{01} + \sqrt{\frac{1-\varepsilon}{3}} \ket{\Psi^+}\ket{10} + \sqrt{\varepsilon} \ket{\Psi^-}\ket{11}
    \end{equation}
    
    Now for the rest of the circuit, a measurement-controlled system will be applied based on the four orthogonal projective measurements.
    If we measured $\ket{00}, \ket{\Psi^+}$ with probabilities $p_1$ and $p_2$, the control system will do nothing and leave the principal system free to output $\sigma_{1,2}$ guaranteed by Schmidt decomposition. Once the state $\ket{11}$ is measured with probability $p_4$, a $U_\varepsilon^\dagger$ will be applied to the main system in order to erase the factor of $\varepsilon$ and produce $\sigma_4$. At last, if $\ket{\Psi^-}$ is measured with probability $p_3$, a unitary $W = (V_{\Psi^-}U^\dagger U_\varepsilon^\dagger\ox I) CV^\dagger$ will be applied on both principal and the ancillary systems to reverse to the initial state $\ket{00}\ket{00}$ and then generate a pure state $\sigma_3 = \ketbra{\Psi^-}{\Psi^-}$.


\section{Comparison with different QPD-based error mitigation methods}\label{appen:comparison_with_diff_methods}
{In this section, we are going to make a comparison of the sampling overhead of three different methods, which are channel inverse, information recover, and observable shift, respectively. The straightforward comparison is shown in the following.}
\begin{figure}[htb]
\label{fig:three_method_table}
    \centering
    \includegraphics[width = 0.9\linewidth]{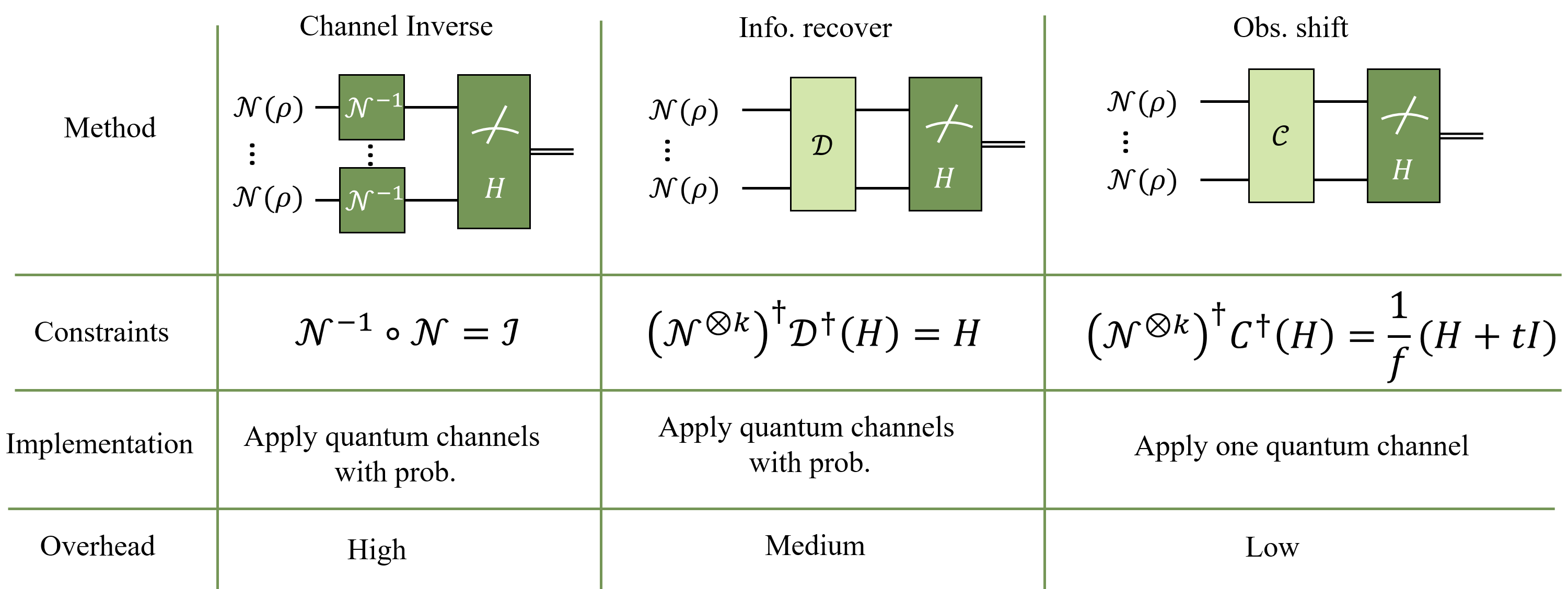}
    \caption{Comparison of the setting between three error mitigation methods.}
\end{figure}

\begin{enumerate}
    \item{ \textbf{Channel inverse.} This method was proposed in Ref.~\cite{temme2017error, jiang2021physical}, which cancels the noise by applying the inverse of the noise channel $\cN^{-1}$ on noisy state $\cN(\rho)$, such that the expectation value is unbiased, i.e., $\tr[O \cN^{-1}\circ\cN(\rho)] = \tr[O\rho]$. The minimal sampling overhead of this method can be given by SDP:
    \begin{align}
        \min & \quad p_1 + p_2\\
        \st & \quad J_\cD \equiv J_1 - J_2\\
        &\quad \tr_{B}[J_1]=p_1 I_{A},\quad \tr_{B}[J_2]=p_2 I_{A}\\
        &\quad J_1 \ge 0, \quad J_2 \ge 0\\
        &\quad J_{\cF_{AC}} \equiv \tr_B[(J_{\cN_{AB}}^{T_B}\ox I_C) (I_A\ox J_{\cD_{BC}})]\\
        &\quad J_{\cF_{AC}} = J_{\id_{AC}}
    \end{align}
    $J_\cD$ and $J_\cN$ are the Choi matrices for retriever $\cD$ and noisy $\cN$ respectively. $J_{\id_{AB}}$ refers to the Choi matrix for the identity channel. }

    \item{ \textbf{Information recover.} This method was proposed in Ref.~\cite{zhao2023information}, which aims to find the hermitian preserving trace preserving map $\cD$ such that the information can be recovered respective to observable $O$, i.e., $\tr[O\cD\circ\cN(\rho)] = \tr[O\rho]$. Note that this method differs from the channel inverse method. The channel inverse can be understood as recovering information for all observables, while this method only recovers information concerning a certain observable. The minimal sampling overhead of this method can be given by SDP:
    \begin{align}
        \min & \quad c_1 + c_2\\
        \st & \quad J_\cD \equiv J_1 - J_2\\
        &\quad \tr_{B}[J_1]=c_1 I_{A},\quad \tr_{B}[J_2]=c_2 I_{A}\\
        &\quad J_1 \ge 0, \quad J_2 \ge 0\\
        &\quad J_{\cF_{AC}} \equiv \tr_B[(J_{\cN_{AB}}^{T_B}\ox I_C) (I_A\ox J_{\cD_{BC}})]\\
        &\quad \tr_C[(I_A\ox O^T_C)J_{\cF_{AC}}^T] = O_A
    \end{align}
    $J_\cD$ and $J_\cN$ are the Choi matrices for retriever $\cD$ and noisy $\cN$ respectively. $T$, $T_B$ stand for transpose and partial transpose, respectively. $O$ refers to observable.}

    \item {\textbf{Observable shift.} This method is proposed in this work. Details can be found in the OBSERVABLE SHIFT METHOD section in the main text. The minimal sampling overhead of this method is given by the following SDP:
    \begin{align}
    \min  &\quad f\\
    \st &\quad J_{\Tilde{\cC}} \ge 0 \\
    &\quad \tr_{C}[J_{\Tilde{\cC}_{BC}}]=f I_{B}\\
    &\quad J_{\cF_{AC}} \equiv \tr_B[(J_{\cN_{AB}^{\ox k}}^{T_B}\ox I_C) (I_A\ox J_{\Tilde{\cC}_{BC}})]\\
    &\quad \tr_C[(I_A\ox H_C^T)J_{\cF_{AC}}^T] = H_A + tI_A.
    \end{align}
    $J_{\Tilde{\cC}}$ and $J_\cN$ are the Choi matrices for retriever $\Tilde{\cC}$ and noisy $\cN$ respectively. $I$ is identity. $T$, $T_B$ stand for transpose and partial transpose, respectively. $H$ refers to observable for estimating high order moment.}
\end{enumerate}

{Here we conduct a numerical simulation to compare the minimal sampling overhead of the three methods in the task of estimating $\tr[\rho^3]$ from noisy states. The noise model is amplitude damping and the observable is $H=(S+S^\dagger)/2$, where $S$ is the cyclic permutation matrix. The results are shown in Fig.~\ref{fig:three_method_AD}. It is straightforward to conclude that the new proposed observable shift method requires the lowest sampling overhead compared with the channel inverse method and observable shift method. Besides, both channel inverse and information recover method requires to implement HPTP map by sampling quantum channels with certain probabilities, while the observable shift method skips that process and only needs to repeat one CPTP map. Thus, the newly proposed method in this work has two advantages: easier implementation and fewer sampling times.}

\begin{figure}[hptb]
\label{fig:three_method_AD}
    \centering
    \includegraphics[width = 0.55\linewidth]{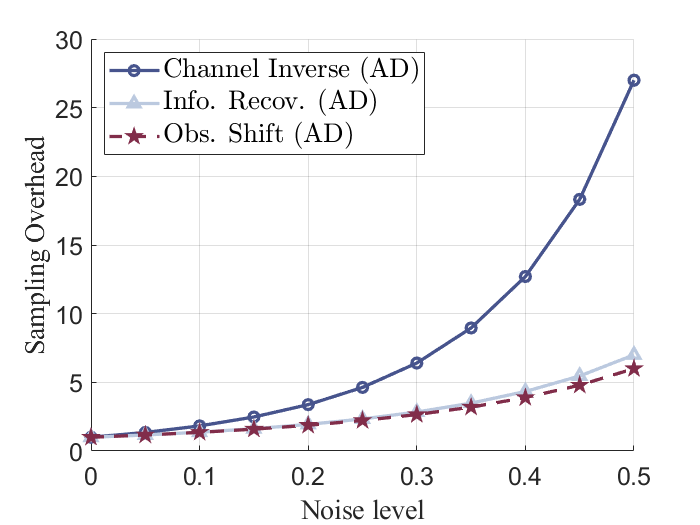}
    \caption{Comparison of sampling overhead between three error mitigation methods under amplitude damping channel.}
\end{figure}

\section{Proof for Lemma ~\ref{lemma:advantage}}

\begin{align}
    f_{\min}(\cN, k) &= \min\Big\{f \mid (\cN^{\ox k})^\dagger\circ \Tilde{\cC}^\dagger(H) = \frac{1}{f} \left( H+tI \right), f\in\mathbb{R}^+, t\in\mathbb{R},  \Tilde{\cC}\in\text{CPTP} \Big\}\\
    g_{\min}(\cN) &= \min\Big\{\sum_i |c_i| \mid \cN^{-1} = \sum_i c_i \cC_i, c_i\in\mathbb{R}, \cC_i\in \text{CPTP}\Big\},
\end{align}

\noindent \textbf{Lemma ~\ref{lemma:advantage}}~
\textit{For arbitrary invertible quantum noisy channel $\cN$, and moment order $k$, we have 
    \begin{equation}
        f_{\min}(\cN, k) \le g_{\min}(\cN, k).
    \end{equation}}

\begin{proof}
    In the task of extracting $k$-th moment from noisy states $\cN(\rho)$, we can apply the inverse operation with optimal decomposition on quantum systems simultaneously to mitigate the error. The corresponding sampling overhead for $k$-th moment will be 
    \begin{equation}
        g_{\min}(\cN, k) = \min\Big\{\sum_i |c_i| \mid (\cN^{-1})^{\ox k} = \sum_i c_i \cC_i,  c_i\in\mathbb{R}, \cC_i\in \text{CPTP}\Big\}.
    \end{equation}

    Let's assume the optimal decomposition with the lowest sampling overhead is 
    \begin{equation}
        (\cN^{-1})^{\ox k} = \sum_i c_i' \cC_i'.
    \end{equation}
    Since $(\cN^{-1})^{\ox k} \circ \cN^{\ox k}= I$, where $I$ refers to identity. we directly have $\cN^{\ox k\dagger} \circ (\cN^{-1})^{\ox k\dagger} (H) = H$, where $H$ is the moment observable. For further step, we have $\cN^{\ox k\dagger} \circ (\cN^{-1})^{\ox k\dagger} (H+tI-tI) = H$, where $t$ is a real coefficient. Note that $\cN^\dagger$ is a CPUP map, and  $(\cN^{-1})^\dagger$ is a HPUP map, thus
    \begin{equation}\label{eq:inverse_transfer}
        \cN^{\ox k\dagger} \circ [(\cN^{-1})^{\ox k\dagger} (H) +tI] = H +tI.
    \end{equation}
    We can denote $\cC^\dagger = [(\cN^{-1})^{\ox k\dagger} (H) +tI]$ can be consider as a whole. If we take the coefficient $t$ greater than the absolute value of the minimum eigenvalue of the hermitian matrix $(\cN^{-1})^{\ox k\dagger} (H)$, which is $t \ge |\min(\text{eig}(\cN^{-1})^{\ox k\dagger} (H))|$. Then, the map $\cC^\dagger$ is a CPUS map, and its corresponding adjoint map $\cC$ is a CPTS map. We can substitute the constructed map $\cC$ back to Eq.~\eqref{eq:inverse_transfer}, we have 
    \begin{equation}
        (\cN^{\ox k})^\dagger\circ \cC^\dagger(H) = H+tI
    \end{equation}
    Since any CPTS map $\cC$ can be written in the form of $\cC = f\Tilde{\cC}$, where $f$ is the real coefficient, and $\Tilde{\cC}$ is a CPTP map. We have
    \begin{equation}
        (\cN^{\ox k})^\dagger\circ f\Tilde{\cC}^\dagger(H) = H+tI,
    \end{equation}
    which has the same form as Eq.~\eqref{eq:f(N,k)}. Therefore, we have proved that the optimal decomposition of Eq.~\eqref{eq:g(N,k)} is one feasible solution of Eq.~\ref{eq:f(N,k)}, and we get the relation $ f_{\min}(\cN, k) \le g_{\min}(\cN, k)$ straightforwardly. The proof is completed
\end{proof}
    
\section{Depolarizing noise retriever for the second moment when \texorpdfstring{$\rho$}{} is \texorpdfstring{$n$}{}-qubit state}\label{appendix:depo n qubit}

\begin{proposition}\label{prop:2nd moment n qubit}
    Given noised states $\cN^\epsilon_{\rm DE}(\rho)^{\ox2}$, and error tolerance $\delta$ where $\rho$ is an $n$-qubit state, the second order moment $\tr[\rho^2]$ can be estimated by $f\tr[H\cC\circ \cN^\epsilon_{\rm DE}(\rho)^{\ox2}] - t$ where $f = \frac{1}{(1-\epsilon)^2}$ and $t = \frac{1-(1-\epsilon)^2}{2^n(1-\epsilon)^2}$. The term $\tr[H\cC\circ \cN^\epsilon_{\rm DE}(\rho)^{\ox2}]$ can be estimated by implementing quantum retriever $\cC$ with optimal sample complexity $\cO(1/\delta^2(1-\epsilon)^4)$. The Choi matrix of such a quantum retriever is,
    $$J_\cC = \frac{1}{4^n} I^{\ox n}I^{\ox n} + \frac{1}{4^n(4^n-1)}(\sum_{i\neq 0} P_i\ox P_i)\ox(\sum_{j\neq 0} P_j\ox P_j).$$
\end{proposition}

\begin{proof}
    For an arbitrary state $\rho\ox\rho$, after applying the noise $\cN$ onto it, we denote the state as $\rho' = \cN(\rho\ox\rho)$, thus 
    \begin{align}
        \cC\circ \cN(\rho\ox\rho) &= \cD(\rho')  = \tr_A[(\rho'^T \ox I^{\ox n}) J_\cC]\\
    &=  \tr_A\Big[(\rho'^T \ox I^{\ox n}) (\frac{1}{4^n} I^{\ox n}I^{\ox n}  + \frac{1}{4^n(4^n-1)}(\sum_{i\neq 0} P_i\ox P_i)\ox(\sum_{j\neq 0} P_j\ox P_j))\Big]\\
    &=\frac{1}{4^n}\tr[\rho'^T] I^{\ox n}
   + \frac{1}{4^n(4^n-1)}\tr[\rho'^T(\sum_{i\neq 0} P_i\ox P_i)](\sum_{j\neq 0} P_j\ox P_j)\\
    &= \frac{1}{4^n}I^{\ox n} +  \frac{1}{4^n(4^n-1)}\tr[\rho'^T(\sum_{i\neq 0} P_i\ox P_i)](\sum_{j\neq 0} P_j\ox P_j)\label{eq:D_circ_N_depo_material}
    \end{align}
    Note that the above equations utilized the face that transpose operation is trace-preserving, i.e., $\tr[\rho^T]=\tr[\rho]=1$. Since the matrix $\sum_{i\neq 0} P_i$ is symmetric, thus we have
    \begin{align}
        \cC\circ \cN(\rho\ox\rho) &= \frac{1}{4^n}I^{\ox n} +  \frac{1}{4^n(4^n-1)} \tr[\rho'^T(\sum_{i\neq 0} P_i\ox P_i)^T](\sum_{j\neq 0} P_j\ox P_j)\\
        &= \frac{1}{4^n}I^{\ox n} +  \frac{1}{4^n(4^n-1)} \tr[\rho'(\sum_{i\neq 0} P_i\ox P_i)](\sum_{j\neq 0} P_j\ox P_j) \label{eq:DE_after_transpose_nq}
    \end{align}
    The trace term can be calculated by substituting Eq.~\eqref{eq:depo_state}, we have 
    \begin{align}
        \tr[\rho'(\sum_{i\neq 0} P_i\ox P_i)] &= \tr[((1-\epsilon)^2 \rho\ox\rho + \frac{\epsilon(1-\epsilon)}{2^n} I^{\ox n}\ox\rho + \frac{(1-\epsilon) \epsilon}{2^n} \rho\ox I^{\ox n} + \frac{\epsilon^2}{4^n} (I^{\ox n}\ox I^{\ox n})) (\sum_{i\neq 0} P_i\ox P_i)]\\
        &= (1-\epsilon)^2\tr[\rho\ox\rho(\sum_{i\neq 0} P_i\ox P_i)\ox(\sum_{j\neq 0} P_j\ox P_j)] .\label{eq:depo_mid_nq}
    \end{align}
    In the second equation, all other terms are traceless, thus, only the first term survives. Replace the equation back to Eq.~\eqref{eq:DE_after_transpose_nq}, then
    \begin{align}
        \cC\circ \cN(\rho\ox\rho) &= \frac{1}{4^n}I^{\ox n} + \frac{(1-\epsilon)^2}{4^n(4^n-1)}\tr[\rho\ox\rho (\sum_{i\neq 0} P_i\ox P_i)] (\sum_{j\neq 0} P_j\ox P_j).
    \end{align}
    The information $\tr[\rho^2]$ is estimated from $ \tr[H(\rho\ox\rho)]$, where $H = \frac{1}{2^n}\sum_{i\neq 0} P_i\ox P_i$ is cyclic permutation operator. It is easy to check that 
    \begin{align}
        \tr[H \cC\circ\cN(\rho\ox\rho)] &=  \frac{1}{4^n}\tr[H*I^{\ox n}] + \frac{(1-\epsilon)^2}{4^n(4^n-1)}\tr[\rho\ox\rho (\sum_{i\neq 0} P_i\ox P_i)]\tr[H*(\sum_{j\neq 0} P_j\ox P_j)].
    \end{align}
    We can quickly get $\tr[H] = 2^n$ and $\tr[H*(\sum_{j\neq 0} P_j\ox P_j)]=2^n(4^n-1)$. Then
    \begin{align}
        f\tr[H \cC\circ\cN(\rho\ox\rho)] &= \frac{1}{(1-\epsilon)^2} \Big[\frac{1}{2^n} +\frac{(1-\epsilon)^2}{2^n}\tr[\rho\ox\rho (\sum_{i\neq 0} P_i\ox P_i)]\Big]\\
        &= \frac{1}{2^n(1-\epsilon)^2} + \frac{1}{2^n}\tr[\rho\ox\rho (\sum_{i\neq 0} P_i\ox P_i)]\\
        &= \frac{1}{2^n (1-\epsilon)^2}-\frac{1}{2^n} + \frac{1}{2^n} +  \frac{1}{2^n}\tr[\rho\ox\rho (\sum_{i\neq 0} P_i\ox P_i)]\\
        &= \frac{1}{2^n(1-\epsilon)^2}-\frac{(1-\epsilon)^2}{2^n(1-\epsilon)^2} +  \frac{1}{2^n}\tr[\rho\ox\rho (\sum_{i} P_i\ox P_i)]\\
        &= \frac{1-(1-\epsilon)^2}{2^n(1-\epsilon)^2} + \tr[\rho\ox\rho H]\\
        &= t + \tr[\rho^2].
    \end{align}
    where $t = \frac{1-(1-\epsilon)^2}{2^n(1-\epsilon)^2}$ is the shift coefficient. The desired high-order moment $\tr[\rho^2]$ is directly obtained by estimated value $\tr[H \cC\circ\cN(\rho\ox\rho)]$ substrate a constant $t$.
\end{proof}
\section{Depolarizing noise retriever for the \texorpdfstring{$k$}{}-th moment when \texorpdfstring{$\rho$}{} is a qudit state}~\label{appendix:depo n qudit k moment}

{Appendix~\ref{appendix:depo n qubit} demonstrates the efficacy of the observable shift method in extracting the second moment of an $n$-qubit quantum state $\rho$ subjected to a depolarizing channel $D_\eps$. This section extends the method's applicability, showing it can also determine the $k$-th moment of a $\rho$ given $k$ copies of $D_\eps(\rho)$ for $100 \geq k > 2$.
Our extension is based on the fact that the $k$-th moment of $\rho$ is intrinsically a linear combination of $\trace{\rho^2}, \ldots, \trace{\rho^{k-1}}$ and $\trace{D_\eps(\rho)^k}$ with a constant, as evidenced by the following binomial expansion:
\begin{equation}\label{eq:depo expansion}
    \trace{D_\eps(\rho)^k} 
    = \trace{\left( (1 - \eps) \rho + \eps I / d \right)^k}
    = \sum_{l = 0}^k \binom{k}{l} \frac{ (1 - \eps)^{l}\eps^{k - l} }{d^{k - l}} \trace{\rho^l}
,\end{equation}
where $I$ is the identity matrix of dimension $d = 2^n$. Then the recovery map can be recursively found by reducing $\trace{\rho^l}$ to an expectation value of Hamiltonian $H_k$ under the depolarizing noise, and finally the observable shift method can be applied. Before presenting the details, we need an extra proposition to guarantee the existence of such a reduction, and the proof of which is deferred to Appendix~\ref{appendix:proof perm cp transfer}.}

\begin{proposition}~\label{prop:cptn permutation transfer}
    {Suppose $100 \geq k > 2$. Denote $H_l = \frac{1}{2} (S_l + S_l^\dag)$ for the cyclic permutation operator $S_l \in \CC^{ld \times ld}$ permuting $l$ subsystems. Then there exists CP maps $\cT_k$ and $\widetilde{\cT}_k$ such that 
\begin{equation}
    \cT_k(H_k) = H_{k - 1} \ox I / d, \quad 
    \widetilde{\cT}_k(H_k) = -H_{k - 1} \ox I / d
.\end{equation}}
\end{proposition}

{To see how this proposition works, for any $k > l \geq 2$, one can construct a CP map
\begin{equation}\label{eq:recovery}
    \cR_l = \left(\cT_{l + 1} \ox {\rm id}_{k - l - 1}\right) \circ \ldots \circ \left(\cT_{k - 2} \ox {\rm id}_2\right) \circ  \left(\cT_{k - 1} \ox {\rm id}\right) \circ \widetilde{\cT}_k
.\end{equation}
By Proposition~\ref{prop:cptn permutation transfer},
\begin{align}
    \cR_l(H_k) &= \left(\cT_{l + 1} \ox {\rm id}_{k - l - 1}\right) \circ \ldots \circ \left(\cT_{k - 2} \ox {\rm id}_2\right) \circ  \left(\cT_{k - 1} \ox {\rm id}\right) \circ \widetilde{\cT}_k(H_k) \\
    &= - \left(\cT_{l + 1} \ox {\rm id}_{k - l - 1}\right) \circ \ldots \circ \left(\cT_{k - 2} \ox {\rm id}_2\right) \left(\cT_{k - 1} \left(H_{k - 1}\right) \ox I / d \right) \\
    &= - \cT_{l + 1}\left(H_{k - l + 1}\right) \ox I_{k - l - 1} / d^{k - l - 1} 
    = - H_l \ox I_{k-l} / d^{k - l}
,\end{align}
where the negative sign here is to preserve completely-positiveness of $\cC_k$ in the later proposition. Then for any state $\sigma$, the computation of $\trace{\sigma^l}$ is understood as an expectation value of $H_k$:
\begin{equation}
    \trace{\sigma^l} = \trace{H_l \sigma^{\ox l}} 
    \propto -\trace{\cR_l(H_k) \sigma^{\ox k}} = -\trace{H_k \cdot \cR_l^\dag \left( \sigma^{\ox k} \right)}
.\end{equation}
We are ready to present the main conjecture in this section.
}

\vspace{0.5em}
\noindent \textbf{Proof of Proposition~\ref{prop:kth moment n qudit}}\,\,
    {The statement is proved by induction.}

{\noindent (\textbf{Base case}: $k = 2$) The base case follows by Proposition~\ref{prop:2nd moment n qubit}.}

{\noindent (\textbf{Inductive hypothesis}) For all $k > l \geq 2$, there exists a CP map $C_l$, constants $f_l$ and $t_l$ such that
\begin{equation}
    f_l \trace{H_l \cdot \cC_l \circ D_\eps^{\ox l} \left(\rho^{\ox l}\right)} - t_l =  \trace{\rho^l}
.\end{equation}}

{\noindent (\textbf{Inductive step})    
Proposition~\ref{lemma:exitance_of_H} and Equation~\eqref{eq:depo expansion} implies
\begin{equation}
    \trace{H_k D_\eps (\rho)^{\ox k} } = \trace{D_\eps (\rho)^{k}}
    = (1 - \eps)^k \trace{\rho^k} + \frac{\eps^k}{d^k} + \frac{k (1 - \eps) \eps^{k - 1}}{d^{k - 1}} + \sum_{l = 2}^{k  - 1} \binom{k}{l} \frac{ (1 - \eps)^{l}\eps^{k - l} }{d^{k - l}} \trace{\rho^{l}}
.\end{equation}
From the induction hypothesis, we have
\begin{equation}
    \trace{H_k \cdot {\rm id}_k \left(D_\eps (\rho)^{\ox k}\right) }
    = (1 - \eps)^k \left(\trace{\rho^k} + t_k \right) +  \sum_{l = 2}^{k  - 1} \binom{k}{l} (1 - \eps)^l\eps^{k - l} f_l \trace{ H_{l} / d^{k - l} \cdot \cC_{l} \left( D_\eps (\rho)^{\ox l} \right)}
,\end{equation}
    where the constant $t_k$ is constructed as
\begin{equation}
    t_k = \frac{1}{(1 - \eps)^k} \left[ \frac{\eps^k}{d^k} + \frac{k (1 - \eps) \eps^{k - 1}}{d^{k - 1}} 
    - \sum_{l = 2}^{k  - 1} \binom{k}{l} \frac{ (1 - \eps)^{l}\eps^{k - l} }{d^{k - l}} t_l \right]
.\end{equation}
    Further, from the construction of $\cR_l$ in Equation~\eqref{eq:recovery}, the trace quantity on the last equality can be expanded as
\begin{align}
     \trace{ H_{l} / d^{k - l} \cdot \cC_{l} \left( D_\eps (\rho)^{\ox l} \right)}
    &= \trace{\left(H_{l} \ox I_{k - l} / d^{k - l} \right) \cdot \left(\cC_{l} \ox {\rm id}_{k - l} \right) \left(D_\eps (\rho)^{\ox k}\right)} \\
    &= -\trace{ \cR_{l}(H_k) \cdot \left(\cC_{l} \ox {\rm id}_{k - l} \right) \left(D_\eps (\rho)^{\ox k}\right)} \\
    &= -\trace{ H_k \cdot \cR_{l}^\dag \circ \left(\cC_{l} \ox {\rm id}_{k - l} \right) \left(D_\eps (\rho)^{\ox k}\right)}
.\end{align}
    By moving the term around, we eventually get
\begin{align}
    \trace{\rho^k} &= \frac{1}{(1 - \eps)^k} \set{\trace{H_k \cdot {\rm id}_k \left(D_\eps (\rho)^{\ox k}\right) } + \sum_{l = 2}^{k  - 1} \binom{k}{l} (1 - \eps)^{l}\eps^{k - l} f_l \trace{ H_k \cdot \cR_{l}^\dag \circ \left(\cC_{l} \ox {\rm id}_{k - l} \right) \left(D_\eps (\rho)^{\ox k}\right)} } - t_k \\
    &= \frac{1}{(1 - \eps)^k} \trace{ H_k \cdot \left({\rm id}_k + \sum_{l = 2}^{k  - 1} \binom{k}{l} (1 - \eps)^{l}\eps^{k - l} f_l \cR_{l}^\dag \circ \left(\cC_{l} \ox {\rm id}_{k - l} \right)\right) \left(D_\eps (\rho)^{\ox k} \right)}  - t_k \\
    &= f_k \trace{ H_k \cdot \cC_k \left(D_\eps (\rho)^{\ox k} \right)}  - t_k
,\end{align}
    where $f_k = 1 / (1 - \eps)^k$
    and
\begin{equation}
    \cC_k = {\rm id}_k + \sum_{l = 2}^{k  - 1} \binom{k}{l} (1 - \eps)^{l}\eps^{k - l} f_l \cR_{l}^\dag \circ \left(\cC_{l} \ox {\rm id}_{k - l} \right)
.\end{equation}
The CP property of $\cC_k$ follows by the positiveness of $f_l$ and the CP properties of $\cC_l, \cR_l$.
\hfill $\blacksquare$}


\subsection{Proof for Proposition~\ref{prop:cptn permutation transfer}}~\label{appendix:proof perm cp transfer}

{This section presents the necessary lemma for proving the existence of physical protocol to retrieve generalized depolarizing noises. Particularly, we show that there exists CP maps that convert the permutation unitary $S_k$ to $\pm S_{k - 1}$ for $k > 2$, as stated in Proposition~\ref{prop:cptn permutation transfer}.
Denote $\o_k = \exp\left(2 \pi i / k \right)$ as the $k$-th unity root. Let $\combwr{k}{d}$ be the minimal set of length-$k$ combinations that generates $\set{0, \ldots, d - 1}^{\times k}$ under cyclic permutations, i.e.,
\begin{equation}
    \set{0, \ldots, d - 1}^{\times k} = \setcond{ x_{\pi^l(1)} \cdots x_{\pi^l(k)} }{ x_{1} \cdots x_{k} \in \combwr{k}{d}, l \in \NN }
.\end{equation}
For example, $\combwr{3}{2}$ can be non-uniquely constructed as $\set{000, 001, 011, 111}$. Note that $\abs{\combwr{k}{d}} = \frac{d^k - d}{k} + d$. One can then have the following result:}

\begin{lemma}[Spectral decomposition of $S_k$]
    {The cyclic permutation matrix can be decomposed as
\begin{equation}
    S_k = \sum_{m = 0}^{k - 1} \o_k^{-m} \sum_{\bm x \in \combwr{k}{d}} \ketbra{ \psi_{m, \bm x} }{ \psi_{m, \bm x} }
,\end{equation}
    for $\bm x = x_1 \cdots x_k$ and eigenstates defined as
\begin{equation}
    \ket{\psi_{m, \bm x}} = \frac{1}{\sqrt{k}} \sum_{l = 0}^{k - 1} \o_k^{ml} \bigotimes_{j = 1}^k \ket{x_{\pi^l(j)}}
.\end{equation}}
\end{lemma}
\begin{proof}
    {We first prove that $\ket{\psi_{m, \bm x}}$ is an eigenstate of $S_k$ with eigenvalue $\o_k^{m}$, then show that the state set forms a basis of $\cH_{d^k}$.
    Applying $S_k$ on $\ket{\psi_{m, \bm x}}$ gives
\begin{align}
    S_k \ket{\psi_{m, \bm x}} &= \left[ \sum_{\bm y} \bigotimes_{j} \ketbra{y_{\pi(j)}}{y_j} \right] \cdot \left[ \sum_l \o_k^{ml} \bigotimes_{j} \ket{x_{\pi^l(j)}} \right] \\
    &= \sum_{\bm y} \sum_l \o_k^{ml} \bigotimes_{j} \delta_{y_j x_{\pi^l(j)}} \ket{y_{\pi(j)}} \\
    &= \sum_{l = 0}^{k - 1} \o_k^{ml} \bigotimes_{j} \ket{x_{\pi^{l + 1}(j)}}
    = \sum_{l = 1}^k \o_k^{m(l - 1)} \bigotimes_{j} \ket{x_{\pi^l(j)}} \\
    &= \o_k^{-m} \ket{\psi_{m, \bm x}}
.\end{align}
Note that for fixed point $\bm x$ such that $x_j = x_{\pi(j)}$ for all $j$, $\ket{\psi_{m, \bm x}} \neq 0$ if and only if $m = 0$. Then the total number of such eigenstates are
\begin{equation}
    d + k \cdot \left( \abs{\combwr{k}{d}} - d \right) = d^k
.\end{equation}}
    {As for the orthogonality, for $0 \leq m, m' < k$, $\bm x, \bm x' \in \combwr{k}{d}$,
\begin{align}
    \braket{ \psi_{m, \bm x} }{ \psi_{m', \bm x'} }
    &= \frac{1}{k} \sum_{l, l'} \o_k^{-ml} \o_k^{m'l'} \prod_{j = 1}^k \braket{ x_{\pi^l(j)} }{ x'_{\pi^{l'}(j)} } \\
    &= \frac{1}{k} \sum_{l, l'} \o_k^{-ml} \o_k^{m'l'} \left( \prod_{j = 1}^k \delta_{x_j x'_{\pi^{(l' - l)}(j)}} \right) \\
    &= \left(\frac{1}{k} \sum_l \o_k^{l(m' - m)} \right) \delta_{{\bm x} {\bm x'}}
    = \delta_{mm'} \delta_{{\bm x} {\bm x'}}
,\end{align}
    where $\bm x \neq \bm x'$ if and only if there does not exists $l$ such that $x_j \neq x'_{\pi^l(j)}$ for all $j$.}
\end{proof}

\vspace{0.5em}
\noindent \textbf{Proof of Proposition~\ref{prop:cptn permutation transfer}}\,\,
    One can observe that for $k > 2$, there exists non-negative matrices $Q, \widetilde{Q} \in \RR_+^{(k-1) \times k}$ such that for all $0 \leq l < k - 2$,
\begin{equation}~\label{eqn:q condition}
    \sum_{m = 0}^{k - 1}  Q_{lm}\o_k^{m}  = \o_{k-1}^l, \quad 
    \sum_{m = 0}^{k - 1} \widetilde{Q}_{lm} \o_k^{m} = -\o_{k-1}^l
.\end{equation}
    Such $Q, \widetilde{Q}$ can construct the desired transformations in Proposition~\ref{prop:cptn permutation transfer}. Using the package CVX~\cite{cvx, gb08}, Equation~\eqref{eqn:q condition} can be numerically verified for $100 \geq k$~\footnote{For the data and verification files, see the GitHub repository.}. For a general $k > 2$, $Q$ and $\widetilde{Q}$ can be explicitly constructed as
\begin{equation}
    {Q}_{\textrm{odd } k} = \begin{pmatrix}
        1 & {} & {} & {} & {} & {} & {} & {} \\
        {} & {x}_1 & {y}_1 & {} & {} & {} & {} & {} \\
        {} & {} & {\ddots} & {\ddots} & {} & {} & {} & {} \\
        {} & {} & {} & {x}_m &{y}_m & {} & {} & {}  \\
        {} & {} & {} & {} & {y}_m & {x}_m & {} & {} \\
        {} & {} & {} & {} & {} & {\ddots} & {\ddots} & {} \\
        {} & {} & {} & {} & {} & {} & {y}_1 & {x}_1
    \end{pmatrix}, \quad
    {Q}_{\textrm{even } k} = \begin{pmatrix}
        1 & {} & {} & {} & {} & {} & {} & {} & {} \\
        {} & {x}_1 & {y}_1 & {} & {} & {} & {} & {} & {} \\
        {} & {} & {\ddots} & {\ddots} & {} & {} & {} & {} & {} \\
        {} & {} & {} &{x}_{m - 1} & {y}_{m - 1} & {} & {} & {} & {}  \\
        {} & {} & {} & {} & {x}_{m} & {x}_{m} & {} & {} & {} \\
        {} & {} & {} & {} & {} & {y}_{m - 1} & {x}_{m - 1} & {} & {} \\
        {} & {} & {} & {} & {} & {} & {\ddots} & {\ddots} & {} \\
        {} & {} & {} & {} & {} & {} & {} & {y}_1 & {x}_1
    ,\end{pmatrix}
\end{equation}
\begin{equation}
    \widetilde{Q}_{\textrm{odd } k} =(X\ox I_{(k-1)/2})\cdot{Q}_{\textrm{odd } k}, \quad
    \widetilde{Q}_{\textrm{even } k} ={Q}_{\textrm{even } k}\cdot(X\ox I_{k/2})
\end{equation}
    dependent on the parity of $k$, where $X = \left(\begin{smallmatrix}
        0 & 1 \\ 1 & 0
    \end{smallmatrix} \right)$ is the Pauli-$X$ gate, $m = \lfloor k / 2 \rfloor$ and constants $x_l, y_l$ are given as
\begin{align}
    x_l &= \csc \frac{2 \pi}{k} \,\sin \frac{2 (k - 1 - l) \pi}{k(k-1)}, \quad
    y_l = \csc \frac{2 \pi}{k} \,\sin \frac{2 l \pi}{k(k-1)}
.\end{align}
    The fact that $k > 2$ ensures the angles exist in $\csc$ are no greater than $\frac{2 \pi}{3}$, and those exist in $\sin$ are no greater than $\frac{\pi}{2}$, and hence these constants are always non-negative.
    
    Now we use $Q, \widetilde{Q}$ to construct the maps.
    {To be specific, denote the eigenspace with respect to $\o_k^{m}$ as $\Pi^{(m)}_k$, i.e., $\Pi^{(m)}_k = \sum_{\bm x \in \combwr{k}{d}} \ketbra{ \psi_{-m, \bm x} }{ \psi_{-m, \bm x} }_k$.
    Consider the map $\cT_k, \widetilde{\cT}_k$ satisfying 
\begin{equation}
    \cT_k \left( \Pi^{(m)}_k \right) = \sum_l Q_{lm} \Pi^{(l)}_{k - 1} \ox I / d, \quad 
    \widetilde{\cT}_k \left( \Pi^{(m)}_k \right) = \sum_l \widetilde{Q}_{lm} \Pi^{(l)}_{k - 1} \ox I / d
\end{equation}
    Then it follows that
\begin{align}
    \cT_k\left( S_k \right) 
    &= \sum_{m = 0}^{k - 1} \o_k^m \cT_k(\Pi^{(m)}_k) = \sum_{m = 0}^{k - 1} \o_k^m \sum_{l = 0}^{k - 2} Q_{lm} \Pi^{(l)}_{k - 1} \ox I / d \\
    &= \sum_l \left(\sum_m  Q_{lm}\o_k^m  \right) \Pi^{(l)}_{k - 1} \ox I / d = \sum_l \o_{k-1}^l \Pi^{(l)}_{k - 1} = S_{k - 1} \ox I / d
,\end{align}
and analogously 
\begin{align}
    \widetilde{\cT}_k \left( S_k \right) 
    & = - S_{k - 1} \ox I / d
.\end{align}
    Since $\set{\Pi^{(m)}_k}_{m = 1}^k$ is an orthogonal set and $Q$ is non-negative, $\cT_k$ can be completely positive and hence $\cT_k\left( S_k^\dag \right) = \cT_k\left( S_k \right)^\dag$. We conclude that
\begin{equation}
    \cT_k\left( H_k \right) = \frac{1}{2} \left[ \cT_k\left( S_k \right) + \cT_k\left( S_k^\dag \right)\right] = \frac{1}{2} \left[ S_{k - 1} + S_{k - 1}^\dag \right] \ox I / d = H_{k - 1} \ox I / d
.\end{equation}
    Similar statement holds for $\widetilde{\cT}_k$.
\hfill$\blacksquare$}

\end{document}

%% file: pretex.tex
\usepackage{mathtools}
\usepackage{amsmath}
\usepackage[shortlabels]{enumitem}

\usepackage{graphicx,epic,eepic,epsfig,amsmath,latexsym,amssymb,verbatim,color}
 
\usepackage{amsfonts}       
\usepackage{nicefrac}       

\usepackage{amsmath}
\usepackage{bbm}

\usepackage{float}
\usepackage{tikz}
\usetikzlibrary{chains}
\usetikzlibrary{fit}
\usetikzlibrary{arrows}

\usepackage{epsfig}
\usetikzlibrary{shapes.symbols,patterns} 
\usepackage{pgfplots}
\pgfplotsset{compat=1.18}

\usepackage[strict]{changepage}
\usepackage{hyperref}
\hypersetup{colorlinks=true,citecolor=blue,linkcolor=blue,filecolor=blue,urlcolor=blue,breaklinks=true}

\usepackage[marginal]{footmisc}
\usepackage{url}
\usepackage{theorem}

\newtheorem{definition}{Definition}
\newtheorem{proposition}{Proposition}
\newtheorem{lemma}[proposition]{Lemma}

\newtheorem{theorem}[proposition]{Theorem}

\def\squareforqed{\hbox{\rlap{$\sqcap$}$\sqcup$}}
\def\qed{\ifmmode\squareforqed\else{\unskip\nobreak\hfil
\penalty50\hskip1em\null\nobreak\hfil\squareforqed
\parfillskip=0pt\finalhyphendemerits=0\endgraf}\fi}
\def\endenv{\ifmmode\;\else{\unskip\nobreak\hfil
\penalty50\hskip1em\null\nobreak\hfil\;
\parfillskip=0pt\finalhyphendemerits=0\endgraf}\fi}
\newenvironment{proof}{\noindent \textbf{{Proof~} }}{\hfill $\blacksquare$}

\newcounter{remark}

\newcounter{example}

\mathchardef\ordinarycolon\mathcode`\:
\mathcode`\:=\string"8000
\def\vcentcolon{\mathrel{\mathop\ordinarycolon}}
\begingroup \catcode`\:=\active
  \lowercase{\endgroup
  \let :\vcentcolon
  }

\usepackage{cleveref}
\usepackage{graphicx}
\usepackage{xcolor}

\RequirePackage[framemethod=default]{mdframed}
\newmdenv[skipabove=7pt,
skipbelow=7pt,
backgroundcolor=darkblue!15,
innerleftmargin=5pt,
innerrightmargin=5pt,
innertopmargin=5pt,
leftmargin=0cm,
rightmargin=0cm,
innerbottommargin=5pt,
linewidth=1pt]{tBox}

\newmdenv[skipabove=7pt,
skipbelow=7pt,
backgroundcolor=blue2!25,
innerleftmargin=5pt,
innerrightmargin=5pt,
innertopmargin=5pt,
leftmargin=0cm,
rightmargin=0cm,
innerbottommargin=5pt,
linewidth=1pt]{dBox}
\newmdenv[skipabove=7pt,
skipbelow=7pt,
backgroundcolor=darkkblue!15,
innerleftmargin=5pt,
innerrightmargin=5pt,
innertopmargin=5pt,
leftmargin=0cm,
rightmargin=0cm,
innerbottommargin=5pt,
linewidth=1pt]{sBox}
\definecolor{darkblue}{RGB}{0,76,156}
\definecolor{darkkblue}{RGB}{0,0,153}
\definecolor{blue2}{RGB}{102,178,255}
\definecolor{darkred}{RGB}{195,0,0}

\newcommand{\nc}{\newcommand}
\nc{\rnc}{\renewcommand}
\nc{\lbar}[1]{\overline{#1}}
\nc{\bra}[1]{\langle#1|}
\nc{\ket}[1]{|#1\rangle}
\nc{\ketbra}[2]{|#1\rangle\!\langle#2|}
\nc{\braket}[2]{\langle#1|#2\rangle}

\nc{\proj}[1]{| #1\rangle\!\langle #1 |}
\nc{\avg}[1]{\langle#1\rangle}
\nc{\rank}{\operatorname{Rank}}
\nc{\smfrac}[2]{\mbox{$\frac{#1}{#2}$}}
\nc{\tr}{\operatorname{Tr}}
\nc{\ox}{\otimes}
\nc{\dg}{\dagger}
\nc{\dn}{\downarrow}
\nc{\cA}{{\cal A}}
\nc{\cB}{{\cal B}}
\nc{\cC}{{\cal C}}
\nc{\cD}{{\cal D}}
\nc{\cE}{{\cal E}}
\nc{\cF}{{\cal F}}
\nc{\cG}{{\cal G}}
\nc{\cH}{{\cal H}}
\nc{\cI}{{\cal I}}
\nc{\cJ}{{\cal J}}
\nc{\cK}{{\cal K}}
\nc{\cL}{{\cal L}}
\nc{\cM}{{\cal M}}
\nc{\cN}{{\cal N}}
\nc{\cO}{{\cal O}}
\nc{\cP}{{\cal P}}
\nc{\cQ}{{\cal Q}}
\nc{\cR}{{\cal R}}
\nc{\cS}{{\cal S}}
\nc{\cT}{{\cal T}}
\nc{\cU}{{\cal U}}
\nc{\cV}{{\cal V}}
\nc{\cX}{{\cal X}}
\nc{\cY}{{\cal Y}}
\nc{\cZ}{{\cal Z}}
\nc{\cW}{{\cal W}}
\nc{\csupp}{{\operatorname{csupp}}}
\nc{\qsupp}{{\operatorname{qsupp}}}
\nc{\var}{{\operatorname{var}}}
\nc{\rar}{\rightarrow}
\nc{\lrar}{\longrightarrow}
\nc{\polylog}{{\operatorname{polylog}}}
\nc{\wt}{{\operatorname{wt}}}
\nc{\av}[1]{{\left\langle {#1} \right\rangle}}
\nc{\supp}{{\operatorname{supp}}}

\nc{\argmin}{{\operatorname{argmin}}}

\def\x{\xi}

\def\o{\omega}

\nc{\RR}{{{\mathbb R}}}
\nc{\CC}{{{\mathbb C}}}
\nc{\FF}{{{\mathbb F}}}
\nc{\NN}{{{\mathbb N}}}
\nc{\ZZ}{{{\mathbb Z}}}
\nc{\PP}{{{\mathbb P}}}
\nc{\QQ}{{{\mathbb Q}}}
\nc{\UU}{{{\mathbb U}}}
\nc{\EE}{{{\mathbb E}}}
\nc{\id}{{\operatorname{id}}}

\nc{\CHSH}{{\operatorname{CHSH}}}

\nc{\be}{\begin{equation}}
\nc{\ee}{{\end{equation}}}
\nc{\bea}{\begin{eqnarray}}
\nc{\eea}{\end{eqnarray}}
\nc{\<}{\langle}
\rnc{\>}{\rangle}
\nc{\rU}{\mbox{U}}

\nc{\ob}[1]{#1}

\nc{\SEP}{{\text{\rm SEP}}}
\nc{\NS}{{\text{\rm NS}}}
\nc{\LOCC}{{\text{\rm LOCC}}}
\nc{\PPT}{{\text{\rm PPT}}}
\nc{\EXT}{{\text{\rm EXT}}}
\nc{\Sym}{{\operatorname{Sym}}}

\nc{\ERLO}{{E_{\text{r,LO}}}}
\nc{\ERLOCC}{{E_{\text{r,LOCC}}}}
\nc{\ERPPT}{{E_{\text{r,PPT}}}}
\nc{\ERLOCCinfty}{{E^{\infty}_{\text{r,LOCC}}}}
\nc{\Aram}{{\operatorname{\sf A}}}

\newcommand{\Choi}{Choi-Jamio\l{}kowski }
\newcommand{\eps}{\varepsilon}

\usepackage{tikz}
\usepackage{hyperref}
\hypersetup{colorlinks=true,citecolor=blue,linkcolor=blue,filecolor=blue,urlcolor=blue,breaklinks=true}

\makeatletter
\def\grd@save@target#1{%
  \def\grd@target{#1}}
\def\grd@save@start#1{%
  \def\grd@start{#1}}
\tikzset{
  grid with coordinates/.style={
    to path={%
      \pgfextra{%
        \edef\grd@@target{(\tikztotarget)}%
        \tikz@scan@one@point\grd@save@target\grd@@target\relax
        \edef\grd@@start{(\tikztostart)}%
        \tikz@scan@one@point\grd@save@start\grd@@start\relax
        \draw[minor help lines,magenta] (\tikztostart) grid (\tikztotarget);
        \draw[major help lines] (\tikztostart) grid (\tikztotarget);
        \grd@start
        \pgfmathsetmacro{\grd@xa}{\the\pgf@x/1cm}
        \pgfmathsetmacro{\grd@ya}{\the\pgf@y/1cm}
        \grd@target
        \pgfmathsetmacro{\grd@xb}{\the\pgf@x/1cm}
        \pgfmathsetmacro{\grd@yb}{\the\pgf@y/1cm}
        \pgfmathsetmacro{\grd@xc}{\grd@xa + \pgfkeysvalueof{/tikz/grid with coordinates/major step}}
        \pgfmathsetmacro{\grd@yc}{\grd@ya + \pgfkeysvalueof{/tikz/grid with coordinates/major step}}
        \foreach \x in {\grd@xa,\grd@xc,...,\grd@xb}
        \node[anchor=north] at (\x,\grd@ya) {\pgfmathprintnumber{\x}};
        \foreach \y in {\grd@ya,\grd@yc,...,\grd@yb}
        \node[anchor=east] at (\grd@xa,\y) {\pgfmathprintnumber{\y}};
      }
    }
  },
  minor help lines/.style={
    help lines,
    step=\pgfkeysvalueof{/tikz/grid with coordinates/minor step}
  },
  major help lines/.style={
    help lines,
    line width=\pgfkeysvalueof{/tikz/grid with coordinates/major line width},
    step=\pgfkeysvalueof{/tikz/grid with coordinates/major step}
  },
  grid with coordinates/.cd,
  minor step/.initial=.2,
  major step/.initial=1,
  major line width/.initial=2pt,
}
\makeatother

\usepackage{thmtools}
\usepackage{thm-restate}
\usepackage{etoolbox}
\makeatletter
\def\problem@s{}
\newcounter{problems@cnt}

\newcommand{\allproblems}{\problem@s}
\makeatother